\documentclass[runningheads,envcountsame]{llncs}

\usepackage[utf8]{inputenc}
\usepackage{amsmath,amssymb}

\usepackage[colorinlistoftodos,prependcaption,disable]{todonotes} %
\usepackage{mathtools}
\usepackage{paralist}
\usepackage{xspace}
\usepackage{lineno}
\usepackage{tikz}
 \usetikzlibrary{trees}
 \usetikzlibrary{shapes}
 \usetikzlibrary{fit}
 \usetikzlibrary{shadows}
 \usetikzlibrary{backgrounds}
 \usetikzlibrary{arrows,automata}
 \usetikzlibrary{calc}

\tikzset{
   n/.style= {circle,fill,inner sep=1.5pt,node distance=2cm}
  ,acc/.style={circle,draw,inner sep=3pt,node distance=2cm}
  ,phantom/.style={circle},
  ,arr/.style={->, >=stealth, semithick, shorten <= 3pt, shorten >= 3pt}
}

\usepackage[appendix=inline]{apxproof} % Keep proofs where they are
\newtheoremrep{lemma}[theorem]{Lemma}
\newcommand{\Proc}{\mathbb{P}}
\newcommand{\Alp}{\Sigma}

\newcommand{\diam}{\mathrm{Diam}}
\newcommand{\A}{\mathcal{A}}
\newcommand{\D}{\mathcal{D}}
\newcommand{\Data}{\mathbb{D}}
\newcommand{\CH}{\mathrm{CH}}

\newcommand{\Lang}{\mathcal{L}}
\newcommand{\arch}{\ensuremath{\mathbb{C}}\xspace}

\newcommand{\pedge}{e^\uparrow}
\newcommand{\cedges}{E^\downarrow}
\newcommand{\pcedges}{E^\updownarrow}
\newcommand{\pedgenew}{e'^\uparrow}
\newcommand{\cedgesnew}{E'^\downarrow}
\newcommand{\pcedgesnew}{E'^\updownarrow}
\newcommand{\swap}{\textsc{swap}}

\newcommand{\move}{\textsc{move}}
\newcommand{\connect}{\textsc{conn}}
\newcommand{\disc}{\textsc{disc}}
\newcommand{\op}{\mathrm{op}}
\newcommand{\Op}{\mathit{Op}}
\newcommand{\nop}{\textsc{nop}}
\newcommand{\swapop}{\swap}
\newcommand{\moveop}{\move}
\newcommand{\connop}{\connect}
\newcommand{\discop}{\disc}
\newcommand{\CA}{\mathbb{C}}
\newcommand{\Arch}{\mathfrak{A}}
\newcommand{\Acc}{\mathit{Acc}}
\newcommand{\tree}{T}
\newcommand{\edgelab}{\ell}
\newcommand{\nset}{[n]}
\newcommand{\nsetp}{\nset^+}
\newcommand{\procfromlab}{\pi}
\newcommand{\neigh}{N}
\newcommand{\TCAset}[2]{\mathcal{T}_{#1}^{#2}}
\newcommand{\TCAs}{\TCAset{\CH}{\Proc}}

\newcommand{\diamtree}{\mathrm{TDiam}}
\newcommand{\traceview}[2]{\mathrm{view}_{#1}(#2)}
\newcommand{\traceparentview}[2]{\mathrm{view}_{#1}^{\uparrow}(#2)}
\newcommand{\maketree}[1]{\mathbb{T}(#1)}

\newcommand{\cc}{\mathit{cc}}
\newcommand{\dc}{\mathit{dc}}
\newcommand{\SyncData}{\mathit{SyncData}}
\newcommand{\sync}{\mathrm{sync}}
\newcommand{\treesync}{T_\sync}
\newcommand{\statefromsync}{\sigma}
\newcommand{\idifroot}[3]{\textsc{id-if-root}_{#1}^{#2}(#3)}
\newcommand{\invdef}[1]{(\texttt{I}_{\textsc{def}}[#1])}
\newcommand{\invstate}[1]{(\texttt{I}_S[#1])}
\newcommand{\invstateparent}[1]{(\texttt{I}_S^\uparrow[#1])}
\newcommand{\invCA}[1]{(\texttt{I}_{\CA}[#1])}
\newcommand{\invtree}[1]{(\texttt{I}_{\tree}[#1])}
\newcommand{\invcc}[1]{(\texttt{I}_{\cc}[#1])}
\newcommand{\invdc}[1]{(\texttt{I}_{\dc}[#1])}

\begin{document}

\title{Distribution of Reconfiguration Languages maintaining Tree-like
Communication Topology
\thanks{Supported by the ERC Consolidator grant D-SynMA (No. 772459).}
}
%\titlerunning{Abbreviated paper title}
% If the paper title is too long for the running head, you can set
% an abbreviated paper title here

\author{Daniel Hausmann%\inst{1}
	\and
Mathieu Lehaut%\inst{1}
\and
Nir Piterman%\inst{1}
}

\authorrunning{D. Hausmann, M. Lehaut, N. Piterman}

\institute{University of Gothenburg and Chalmers University of
Technology, Sweden}
%\email{lehaut@chalmers.se}

\maketitle

\begin{abstract}
We study how to distribute trace languages in a setting where processes
communicate via reconfigurable communication channels.
That is, the different processes can connect and disconnect from
channels at run time.
We restrict attention to communication via tree-like communication
architectures.
These allow channels to connect more than two processes in a way that
maintains an underlying spanning tree and keeps communication
continuous on the tree.
We make the reconfiguration explicit in the language allowing both a
centralized automaton as well as the distributed processes to share
relevant information about the current communication configuration.
We show that Zielonka's seminal result regarding distribution of
regular languages for asynchronous automata can be  generalized
in this setting, incorporating both reconfiguration and more than
binary tree architectures.
\end{abstract}

\section{Introduction}
Over the past years, computing devices have proliferated to the
point where they are ubiquitous.
These devices are not only lightweight and affordable but also highly
prevalent and mobile.
They are used, for example, in sensor networks or multi-agent or
multi-robot systems.
In such applications, communication requires mobility and ad-hoc
connectivity depending on need and availability.
Indeed, participants in communications frequently leave or join specific
tasks or groups and change who they communicate with based on, e.g.,
requirements or location.
This prompts us to explore how
communication dynamics evolve when dealing with numerous mobile
participants that have to collaborate effectively.

We consider an extension of Zielonka's asynchronous automata~--~a
canonical formalism for distributed systems with a fixed communication
architecture~\cite{Zielonka87}.
A system consists of multiple finite-state automata called processes.
Processes communicate by synchronizing on channels.
While in asynchronous automata processes share their full state
information in every communication, here, we replace this with
information sharing through data exchange.
The data exchange allows to define how
processes update their local states without knowing the number of
participants in a communication or their identity.
A communication happens only if all involved participants are ready for
it.
Thus, a single process can block communications they are not ready to
accept.
The considered extended asynchronous automata allow processes to decide in every
state to which channels they wish to be connected to.
Thus, their connectivity changes from transition to transition at run
time.
They can easily establish groups and change the set of processes they
communicate with.
This extension is called \emph{reconfigurable asynchronous automata}
(RAA), and was first defined in \cite{LehautP24}.
RAA are inspired by the attribute-based communication calculus \cite{alrahman2015calculus,abd2019calculus}
and by channeled transition systems \cite{AlrahmanP21}, but are simplified to
include only one form of communication (CTS include both broadcast and
multicast channels, the latter corresponding to the type of
communication
allowed here).

Asynchronous automata are notable due to Zielonka's seminal result on
\emph{distributivity} of regular languages \cite{Zielonka87}.
Namely, for every regular language that respects the independence of
communications in a given fixed communication architecture, it is
possible to define a team of processes who distribute the original
language.
That is, their distributed interaction reproduces the regular language.
In other words, one can build a \emph{distributed} implementation from a \emph{centralized} specification.
Zielonka's result is quite involved and is source for much ongoing
research work on establishing a better understanding of the result \cite{MukundS97,genest2006constructing,GenestGMW10,gimbert2022distributed,adsul2024expressively}.
Here we concentrate on the restriction of this result to the case where
processes communicate via an acyclic binary communication architecture~\cite{KrishnaM13}.
That is, every channel is binary (between two processes) and the graph
of channels forms a tree.
Distribution in this restricted setting is much better understood:
the tree enables a distribution
construction where the state space of every process is only quadratic
in the state space of the original deterministic automaton that
(centrally) recognizes the language.
In comparison, in the general case,
the distribution construction
incurs exponential
blow-up (with a non-linear exponent).

As mentioned, in asynchronous automata every channel has a fixed set of
processes connected to it, while our framework allows processes to
connect and
disconnect to/from channels.
We call this process \emph{reconfiguration}, as processes reconfigure
their communication interfaces, and we call the way they are connected
at a certain point in time the \emph{communication architecture}.
In order to extend Zielonka's distribution
results to our automata we make reconfiguration explicit in the
communication.
This allows both the centralized and the distributed systems to follow
how the communication architecture changes.
We do that by including, with every communication, instructions
regarding how the communication architecture is changed and do not allow
other (invisible) changes to the architecture.
Based on this information, the centralized automaton can keep track of
the communication architecture.
We make this notion explicit by considering deterministic automata that
keep track of both the current run of the automaton,
as well as the changing communication architecture.
The instructions that we include allow changes in the structure of a
tree by reversing the order between parent and child vertices and by
transferring a child vertex to its parent or to another child.
In addition, we allow joining or leaving channels in restricted ways
that preserve the connectivity as \emph{tree-like}.
A communication architecture is tree-like if it has an underlying spanning
tree and allows communication that is local with respect to this tree.
That is, all communication is continuous on the tree.
We show that the allowed reconfigurations enable transfer
between any two tree-like communication architectures (in multiple
steps).

Based on these definitions of reconfigurations and reconfiguration
languages, we show that regular languages of reconfigurations can be distributed
if they satisfy a reasonable structure constraint corresponding to the independence of
communications.
That is, they can be recognized by reconfigurable asynchronous automata
that connect to channels according to the instructions given to them
and are only aware of communications they participate in.
To this end, we generalize the construction of Krishna and Muscholl
\cite{KrishnaM13} by adding to it
a part that follows distributively the communication architecture.
We show that these additional parts allow to keep track of how to
update the local view of the global state.
They allow to do so also in the context of communication that is not
necessarily binary (as mentioned, we support tree-like communication
architectures in which channels may subsume more than two participants
provided they are continuous on the spanning
tree).
This, as far as we know, has not been recognized even for fixed
communication structures.
The distribution constructs processes whose number of states is
polynomial in the number of states of a centralized automaton but also keeps track of potentially exponentially many options for the way the
communication architecture is arranged around a process.

The paper is organized as follows.
In Section~\ref{sec:automata} we define automata, asynchronous
automata, and reconfigurable asynchronous automata.
In Section~\ref{sec:tree-like} we define tree-like
communication architectures and how they change.
Section~\ref{sec:reconfig langs} shows how
reconfiguration is exposed through the communication itself, and defines
requirements on languages that enable distribution.
We show in Section~\ref{sec:diamond closed} the distribution
construction and its correctness and conclude in Section~\ref{sec:conc}.

\section{Automata and Reconfigurable Asynchronous Automata}\label{sec:automata}

%We are generally interested in distribution of computation done by
%finite automata.
%In order to make this explicit, we use the term \emph{channels} rather
%than \emph{alphabet letters} and \emph{communication} rather than
%\emph{word}.

\paragraph{Deterministic Finite Automata.}

Given a set of letters $\Sigma$, a deterministic finite automaton
is $\D = (S, s^0, \Delta, F)$, where
$S$ is a finite set of states, $s^0\in S$ is an initial state,
$\Delta:S\times \Sigma \to S$ is a partial transition
function, and $F\subseteq S$ is a set of accepting states.
Given a word $w=a_0\cdots a_{n-1}$, a run of $A$ on $w$ is
$r=s_0,\ldots, s_n$ such that $s_0=s^0$
and for every $0\leq i < n $ we have $s_{i+1} = \Delta(s_i,a_i)$.
A run is accepting if $s_n \in F$ and then $w$ is accepted by
$\D$.
The language of $\D$, denoted by $\Lang(\D)$, is the set of
words accepted by $\D$.

An \emph{independence relation} is a symmetric relation $I \subseteq
\Sigma\times \Sigma$.
%The pair $(\Alp,\dom)$ is called a distributed alphabet.
%We let $\dom^{-1}(p) = \{a \in \Alp \mid p \in \dom(a)\}$.
%It induces an independence binary relation $I$ in the following way:
%%%$(a,b) \in I \Leftrightarrow \dom(a) \cap \dom(b) = \emptyset$.
Two words $u = u_1 \cdots u_n$ and $v = v_1 \cdots v_n$ are said
to be $I$-indistinguishable, denoted by $u \sim_I v$, if one can start
from $u$, repeatedly switch two consecutive independent letters, and
end up with $v$.
That is, for all $(a,b)\in I$ and all $u,v\in \Sigma^*$, we
have $uabv\sim_I ubav$.
We denote by $[u]_I$ the $I$-equivalence class of a word $u\in\Sigma^*$.
Let $\A = (S,s_0,\Delta,F)$ be a deterministic automaton over
$\Sigma$.
We say that $\A$ is $I$-diamond if for all pairs of independent
letters $(a,b) \in I$ and all states $s \in S$, we have $\Delta(s,ab)
= \Delta(s,ba)$.
If $\A$ has this property, then a word $u$ is accepted by $\A$ if and
only if all words in $[u]$ are accepted.
%Zielonka's result states that an $I$-diamond automaton can be
%distributed in certain ways; we will state this result formally after
%introducing asynchronous
%automata \cite{Zielonka87} below.

\paragraph{Asynchronous Automata.}
Let $\Proc$ be a finite set of \emph{processes}.
A \emph{communication architecture} is
$\arch: \Sigma \to 2^\Proc$, associating with each letter the
subset of processes reading it.
We put $\arch^{-1}(p) = \{a \in \Sigma \mid p \in \arch(a)\}$.
A distributed alphabet is $(\Sigma,\arch)$, where $\arch$ is a
communication architecture.
It induces an independence relation
$I(\arch) \subseteq \Sigma\times\Sigma$ by $(a,b)\in I(\arch)$ iff
$\arch(a)\cap
\arch(b)=\emptyset$.
We say that $\arch$ is binary if for each $a\in \Sigma$ we have
$|\arch(a)|=2$.
A binary $\arch$ induces a tree if the graph $(\Proc,E)$, where
$(p,q)\in E$ iff $\arch(a)=\{p,q\}$ for some $a\in \Sigma$, is a tree.

An \emph{asynchronous automaton} (in short: AA) \cite{Zielonka87} over
distributed alphabet $(\Sigma,\arch)$ and processes $\Proc$ is
$\mathcal{B} = ((S_p)_{p \in \Proc}, (s^0_p)_{p \in \Proc}$,
$(\delta_a)_{a \in
	\Alp}, \Acc)$
such that:
\begin{compactitem}
	\item $S_p$ is the finite set of states for process $p$, and $s^0_p
	\in S_p$ is its initial state,
	\item $\delta_a: \prod_{p \in \arch(a)} S_p \to \prod_{p \in
	\arch(a)}
	S_p$ is a partial transition function for letter $a$ that
	only depends on the states of processes in $\arch(a)$ and changes
	them,
	\item $\Acc \subseteq \prod_{p \in \Proc} S_p$ is a set of accepting
	states.
\end{compactitem}
A global state of $\mathcal{B}$ is  $\textbf{s} = (s_p)_{p
\in \Proc}$, giving the state of each process.
For global state $\textbf{s}$ and subset $P \subseteq \Proc$, we
denote by $\textbf{s}{\downarrow}_P = (s_p)_{p \in P}$ the subset of
$\textbf{s}$ of states from processes in $P$.
A run of $\mathcal{B}$ on a word $a_1 a_2\ldots a_n$ is a sequence $\textbf{s}_0 a_1
\textbf{s}_1 a_2
\dots \textbf{s}_n$ such that for all $0 < i \leq n$, $\textbf{s}_i \in
\prod_{p \in \Proc} S_p$, $a_i \in \Alp$, satisfying $\textbf{s}_0 =
(s^0_p)_{p \in \Proc}$ and the following relation:
%\[
$\textbf{s}_{i}{\downarrow}_{\arch(a_i)} =
\delta_{a_i}(\textbf{s}_{i-1}{\downarrow}_{\arch(a_i)})$
 and
$\textbf{s}_{i}{\downarrow}_{\Proc\setminus\arch(a_i)} =
\textbf{s}_{i-1}{\downarrow}_{\Proc\setminus\arch(a_i)}
$.
%\]
A run is accepting if $\textbf{s}_n$ belongs to $\Acc$.
The word $a_1 a_2 \dots$ is accepted by $\mathcal{B}$ if such an
accepting run
exists (note that automata are deterministic but runs on certain words
may not exist).
The language of $\mathcal{B}$, denoted by $\Lang(\mathcal{B})$, is the
set of words
accepted by $\mathcal{B}$.

We say that a language $\Lang$ over
$(\Sigma,\arch)$ is \emph{distributively recognized} if
there exists an asynchronous automaton $\mathcal{B}$
such that $\Lang(\mathcal{B})=\Lang$.

\begin{theorem}[\cite{Zielonka87,KrishnaM13}]
	Given an $I(\arch)$-diamond deterministic automaton $\D$, there
	exists an AA
	that distributively recognizes this language.
	In general if $\D$ has $n$ states then every process in the AA has
	$2^{O(n^2)}$ states.
	If $\arch$ induces a tree, then every process in the AA has $O(n^2)$
	states.
\end{theorem}

\paragraph{Reconfigurable Asynchronous Automata.}

We present Asynchronous Automata with
reconfiguration.
Here the communication architecture is modified dynamically during a
run.
We fix a finite set $\Proc$ of processes %, with $|\Proc| = n$,
and a finite set $\CH$ of channels (intuitively corresponding to the
alphabet).
%Reconfigurable Asynchronous Automata are inspired by \emph{channeled
%transition systems} \cite{AlrahmanP21},
%though they are simpler than those.\dhnote{in which way?}

Let $\Data$ be a finite set of synchronization data contents.
A \emph{reconfigurable asynchronous automaton} (RAA) over
alphabet $\CH$ and processes $\Proc$ is
$\mathcal{A} = ((S_p)_{p \in \Proc}, (s^0_p)_{p \in \Proc}$,
$(\delta_p)_{p \in 	\Proc}, (L_p)_{p\in \Proc},(F_p)_{p\in \Proc})$,
where
\begin{compactitem}
	\item $S_p$ and $s^0_p$ are as before,
	\item $\delta_p: S_p \times (\Data \times \CH) \rightharpoonup S_p$ is the
	partial
	transition function, where $\delta_p(s,(d,c)) = s'$ means going from
	state $s$ to $s'$, reading a data on channel $c$ with content
	$d$.
	We write $(s,(d,c),s')\in \delta$ for $\delta(s,(d,c))=s'$,
	\item $L_p: S_p \to 2^\CH$ is a listening function such that $c \in
	L_p(s)$
	if
	there is a transition of the form $(s,(d,c),s') \in \delta_p$, i.e.
	state
	$s$ must be listening to channel $c$ if there is some transition from
	$s$ involving a data on $c$.
	\item $F_p:S_p \times \Data \rightarrow \{\bot,\top\}$ is an acceptance
	function.
\end{compactitem}
Global states of $\A$ are as before.
A run of $\A$ is $\textbf{s}_0 m_1
\textbf{s}_1 m_2
\dots \textbf{s}_n$, where for all $0 < i \leq n$,
$\textbf{s}_i=(s^i_p)_{p\in \Proc}$ is a global state, $m_i=(d_i,c_i)
\in \Data \times
\CH$, satisfying
$\textbf{s}_0 =
(s^0_p)_{p \in \Proc}$ and all the following:
	\begin{compactenum}
	\item $\exists p$ s.t. $c_i \in L_p(s_p^{i-1})$,
	\item $\forall p$ s.t. $c_i \in L_p(s_p^{i-1}), (s^{i-1}_p, (d_i,c_i), s^i_p) \in
	\delta_p$, and
	\item $\forall p$ s.t. $c_i \notin L_p(s_p^{i-1}), s^i_p = s^{i-1}_p$.
\end{compactenum}
In plain words, there is a transition upon reading channel name $c$
if all processes listening to $c$ have an according transition,
with at least one process listening to $c$, whereas those processes
that do not listen to $c$ are left unchanged.
Note that if some process listens to $c$ but does not implement the
transition, then that transition is blocked.
A run is accepting if for some $d \in \Data$ and for all $p\in \Proc$ we
have $F_p(s^p_n,d)=\top$.
That is, acceptance happens only if all processes can agree on a data value $d$
that is accepted by all.
The language of $\A$, denoted by $\Lang(\A)$, is the set of
words over $\CH$ of the form $c_0 c_1 \dots$ such that
there exists an accepting run of the form $\textbf{s}_0 (d_0, c_0) \textbf{s}_1 (d_1, c_1) \dots$, i.e. we focus only on the
sequence of channels where data is sent, and drop the states and data
contents.

%An RAA over $\CH$ and $\Proc$ refers to an RAA of
%the form $\A_{\parallel \Proc}$ as described
%above.

\begin{example}
	Figure~\ref{example:CTS} shows an example RAA over channels
	$\CH = \{a,b,c\}$ and three processes $\Proc =
	\{p,q,r\}$, with the listening function given by the labels to the right of the individual states.
	Here we take a singleton set $\Data = \{d\}$ as the set of data contents, so for
	readability purposes it is omitted from the
	transitions.
	Note that when process $p$ is in state $s_2$, it is listening to channel $c$ but no $c$-transition is
	implemented, therefore a communication on $c$ is impossible (similarly for $q$ and $t_2$).
	So the only way a communication can happen on $c$ is when $p$ and $q$ are in $s_1$ and $t_1$ respectively, which
	means only process $r$ listens to $c$. %\ml{Missing definition of $F_p$ and description of the language accepted}
	We then set $F_p(x,d) = \top$ for all $p \in \Proc$ and $x \in S_p$.
	The language accepted by this RAA is the set of all sequences such that in any prefix ending with a communication on $c$, we have an even number of $a$ and $b$:
$$\Lang(\A) = \left \{ v_0\ldots v_n \in \{a,b,c\}^*  \left |~
\begin{array}{r}
\forall i~.~ v_i=c \mbox{ implies } a_\sharp(v_0\ldots v_i) =_{mod 2} 0
\\
\mbox{ and }
b_\sharp(v_0\ldots v_i) =_{mod 2}0
\end{array}
\right .
\right \},$$
\noindent
where $\sigma_\sharp(w)$ is the number of occurrences of letter
$\sigma$ in word $w$.
%	It is then easy to see that this RAA accepts the same language as the
% AA given in Figure~\ref{example:AA}.
%	Note that it does so without $p$ or $q$ ever taking part in a
% communication on $c$, contrary to the previous
%	example.
\end{example}

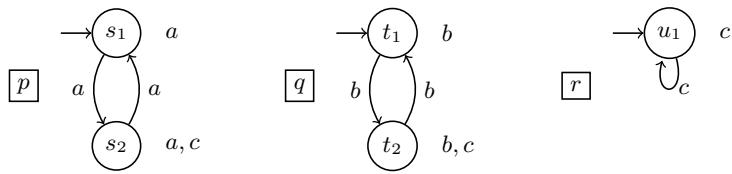
\begin{figure}[bt]
	\tikzset{every state/.style={minimum size=15pt}}
	\begin{center}
		\begin{tikzpicture}[
			% Default arrow tip
			%-&gt;,&gt;=stealth',shorten &gt;=1pt,
			auto,
			% Default node distance
			node distance=1.5cm,
			% Edge stroke thickness: semithick, thick, thin
			semithick
			]
			\node[state,initial left,initial text=] (a1) {$s_1$};
			\node[] (La1) [right = 0.2cm of a1] {$a$};
			\node[draw,rectangle] (p) [left = 0.7cm of a1,yshift=-7mm] {$p$};
			\node[state] (a2) [below of=a1] {$s_2$};
			\node[] (La2) [right = 0.2cm of a2] {$a,c$};
			\node[state,initial left,initial text=] (b1) [right = 3cm of a1]
			{$t_1$};
			\node[] (Lb1) [right = 0.2cm of b1] {$b$};
			\node[draw,rectangle] (q) [left = 0.7cm of b1,yshift=-7mm] {$q$};
			\node[state] (b2) [below of=b1] {$t_2$};
			\node[] (Lb2) [right = 0.2cm of b2] {$b,c$};
			\node[state,initial left,initial text=] (c1) [right = 3cm of b1]
			{$u_1$};
			\node[] (Lc1) [right = 0.2cm of c1] {$c$};
			\node[draw,rectangle] (r) [left = 0.7cm of c1,yshift=-7mm] {$r$};

			\path[->] (a1) edge [bend right] node [left] {$a$} (a2);
			\path[->] (a2) edge [bend right] node [right] {$a$} (a1);
			\path[->] (b1) edge [bend right] node [left] {$b$} (b2);
			\path[->] (b2) edge [bend right] node [right] {$b$} (b1);
			\path[->] (c1) edge [loop below] node [right] {$c$} (c1);
		\end{tikzpicture}
		\caption{An RAA $\A$ over three processes.}\label{example:CTS}
	\end{center}
\vspace*{-7mm}
\end{figure}

We say that an RAA over $\CH$ and $\Proc$ has a fixed communication
architecture if for all $p\in \Proc$ and
all $s,s'\in S_p$ we have $L_p(s)=L_p(s')$.
That is, the listening function of each process never changes.
In that case, there is an obvious definition of a communication
architecture $\arch$.
%We can define the \emph{communication architecture} as
%$\arch: \CH \to 2^\Proc$, associating with each channel the
%subset of processes listening to it in the obvious way.
%Given a set of processes $\Proc$ and a communication architecture
%$\arch$,
%we put $\arch^{-1}(p) = \{a \in \CH \mid p \in \arch(a)\}$.

%Lehaut and Piterman have shown that the direct definition of
%asynchronous automata and the indirect definition through RAA are
%equivalent \cite{LehautP24}.
%That is, starting from RAA with fixed communication architecture we
%can
%construct classic asynchronous automata with exactly
%the same structure, where the complexity of messages and the local
%transitions is encoded in the (classic) global transition.
%In the other direction, the complexity of the (classic) global
%transition is encoded in the messages and the local transitions.

Lehaut and Piterman showed that RAA with fixed architecture and
asynchronous automata are different notations for the same object.

\begin{lemma}[\cite{LehautP24}]
There is an isomorphism between RAA with fixed communication
architecture and asynchronous automata.
%	Every language recognized by an AA over $(\Alp,\dom)$ and $\Proc$ can be recognized
%	by an RAA with set of channels $\Alp$ and processes $\Proc$.
	\label{lemma:AA to CTS}
\end{lemma}

%Starting from RAA with fixed communication architecture we can
%an asynchronous automaton with exactly
%the same structure, where the complexity of data and the local
%transitions is encoded in the global transition.
%In the other direction, the complexity of the global
%transition is encoded in the data and the local transitions.

%In the case that for every $c\in \CH$ we have $|\mathit{dom}(c)|
%\leq 2$,\dhnote{$\mathit{dom}$ undefined at this point}
%then an AA is binary.
%The communication architecture of an AA is a \emph{tree} if the mapping
%$T_{\arch} = (\Proc,r,E,\edgelab)$, where $(p,q)\in E$ iff
%$\{p,q\} \subseteq \arch(c)$ for some $c$, is indeed a tree.

%The communication architecture $\arch$ induces an independence
%binary relation $I$ in the following way:
%$(a,b) \in I \Leftrightarrow \dom(a) \cap \dom(b) = \emptyset$.

\section{Tree-like communication architecture with reconfiguration}
\label{sec:tree-like}

We have defined what it means for a communication architecture to
be a tree.
Now we extend this in two orthogonal directions.
First, we are interested in a more general communication architecture
that is \emph{tree-like} rather than a tree.
Second, we consider the fact that for RAA (unlike AA) the communication
architecture can change over time.
This is made formal below.

%We want to support channels that can involve more than two
%participants and the possibility of reconfiguration while still
%adhering to a tree-like shape for the communication architecture.

\subsection{Tree-like communication architecture}

As before, let $\Proc$ be a set of processes and $\CH$ a set of
channels, both finite.
A communication architecture is tree-like if there is a tree that spans
channels that connect more than two processes.
Formally, we have the following.

Let $|\Proc| = n$ and let us denote by $\nset$ the set
$\{0,...,n-1\}$,
and put $\nsetp = \nset \setminus \{0\}$.
A \emph{tree} over $\Proc$ is a tree $\tree =
(\Proc,r,E,\edgelab)$ where $\Proc$ is the set of nodes, of which $r
\in \Proc$ is the \emph{root}, $E \subseteq \Proc^2$ is the set of
\emph{edges} (so that $|E|=n-1$) that is cycle free, and $\edgelab: E
\to
\nsetp$ is a \emph{labeling bijection} that maps each edge to a
strictly positive integer.

%Let $|\Proc| = n$ and let us denote by $\nset$ the set
%$\{0,...,n-1\}$, and put $\nsetp = \nset \setminus \{0\}$.
%
%A \emph{communication architecture} is a sequence $\CA = ((\Proc_c)_{c
%\in \CH})$ where $\Proc_c$ represents the set of processes that are
%currently listening to channel $c$.
%A \emph{spanning tree} over $\Proc$ is a tree $\tree =
%(\Proc,r,E,\edgelab)$ where $\Proc$ is the set of nodes, of which $r
%\in \Proc$ is the \emph{root}, $E \subseteq \Proc^2$ is the set of
%\emph{edges} (of size $n-1$), and $\edgelab: E \to \nsetp$ is a
%\emph{labeling bijection} that maps each edge to a strictly positive
%integer.

\begin{definition}\label{def:TCA}
We say that $\Arch = (\CA,\tree)$, where $\CA$ is a communication
architecture and $\tree$ is a tree, is a \emph{tree-like communication
architecture} (short: TCA) if:
\begin{compactenum}
\item\label{item:morethantwo} $|\arch(c)| \geq 2$ for all channels $c
\in \CH$,
\item\label{item:channelconnectivity} for all channels $c \in \CH$, the
set $\arch(c)$ is connected in $T$, i.e. if $p,q \in \arch(c)$, then
all processes along the path from $p$ to $q$ in $T$ are also in
$\arch(c)$,
and
\item\label{item:edgechannel} if $(p,q) \in E$, then there is a channel
$c\in \CH$ such that $p,q \in \arch(c)$.
\end{compactenum}
\end{definition}
Condition~\ref{item:morethantwo} prevents a channel from being trivial.%
\footnote{We could allow channels with only $1$ participant with few modifications.
A channel with $0$ participants will stay at $0$ forever, thus being useless to the computation.}
Condition~\ref{item:channelconnectivity} ensures that a communication on a given channel is always ``local'' in the tree.
Finally, condition~\ref{item:edgechannel} ensures that the tree does not consist of disconnected parts.
Let $\TCAset{\CH}{\Proc}$ be the set of all possible TCA over the
sets $\CH$ of channels and $\Proc$ of processes.

\begin{example}
Fix sets $\Proc = \{p_1,\dots,p_5\}$ of processes and $\CH = \{c_1, c_2, c_3\}$ of communication channels.
Then $\Arch = ((\arch(c_1), \arch(c_2), \arch(c_3)), \tree)$ as given
in Figure~\ref{fig:archi} is a tree-like communication architecture
rooted in $p_1$:
\begin{compactenum}
\item All channels have at least 2 members,
\item Every channel is connected in $\tree$,
\item All edges in $\tree$ are covered by at least one channel.
\end{compactenum}

\end{example}

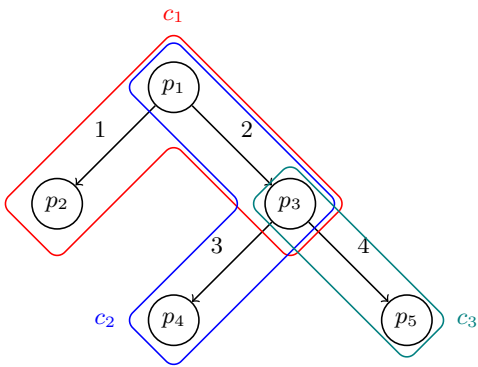
\begin{figure}[bt]
\tikzset{every state/.style={minimum size=15pt}}
\begin{center}
  \begin{tikzpicture}[
    % Default arrow tip
    %-&gt;,&gt;=stealth',shorten &gt;=1pt,
		auto,
    % Default node distance
    node distance=1.5cm,
    % Edge stroke thickness: semithick, thick, thin
    semithick
    ]
    \node[state] (1) {$p_1$};
    \node[state,below left = of 1] (2) {$p_2$};
    \node[state,below right = of 1] (3) {$p_3$};
    \node[state,below left = of 3] (4) {$p_4$};
    \node[state,below right = of 3] (5) {$p_5$};

%    \tikzset{
%      contour1/.style={red,double,double distance=9mm,cap=round,},
%      contour2/.style={blue,double,double distance=10mm,cap=round,},
%      contour3/.style={green,double,double distance=11mm,cap=round,}
%    }

    \path[->] (1) edge node[above left] {$1$} (2);
    \path[->] (1) edge node[above right] {$2$} (3);
    \path[->] (3) edge node[above left] {$3$} (4);
    \path[->] (3) edge node[above right] {$4$} (5);

%    \begin{pgfonlayer}{bg}
%      \draw[contour1] (2.center) --  (1.center) -- (3.center);
%    \end{pgfonlayer}
%    \begin{pgfonlayer}{bg}
%      \draw[contour2] (1.center) -- (3.center) -- (4.center);
%    \end{pgfonlayer}
    \draw[rounded corners,color=red] ($(1.north)+(0,0.4)$) -- ($(2.west)+(-0.4,0)$) -- ($(2.south)+(0,-0.4)$) -- ($(1.south)+(0,-0.4)$) -- ($(3.south)+(0,-0.4)$) -- ($(3.east)+(0.4,0)$) -- cycle node[above,color=red] {$c_1$};
    \draw[rounded corners,color=blue] ($(1.west)+(-0.3,0)$) -- ($(3.west)+(-0.3,0)$) -- ($(4.west)+(-0.3,0)$) node[left,color=blue] {$c_2$} -- ($(4.south)+(0,-0.3)$) -- ($(3.east)+(0.3,0)$) -- ($(1.north)+(0,0.3)$) -- cycle;
    \draw[rounded corners,color=teal] ($(3.north)+(0,0.2)$) -- ($(3.west)+(-0.2,0)$) -- ($(5.south)+(0,-0.2)$) -- ($(5.east)+(0.2,0)$) node[right,color=teal] {$c_3$} -- cycle;
  \end{tikzpicture}
\caption{A tree-like communication architecture: the spanning tree is
given by the black edges, while the communication architecture is drawn
with one color group for each channel. Edges labels are drawn next to
each edge with a matching color.}\label{fig:archi}
\end{center}
\vspace*{-1cm}
\end{figure}

We fix notation regarding the spanning tree $\tree = (\Proc,r,E,\edgelab)$ of a TCA that will be useful later on.
For $p \neq r$, we let $\pedge(p) \in E$ denote the parent edge of $p$; it is the only edge $(p,q)$ (or $(q,p)$) that leads from $p$ toward the root $r$.
We also refer to $q$ as the parent of $p$.
We then let $\cedges(p) \subseteq E$ denote the set of children edges of $p$, which are all edges that $p$ is part of, except for $\pedge(p)$.
Similarly, processes that can be found by following edges from $\cedges(p)$ are the children of $p$.
Finally, put $\pcedges(p) = \{\pedge(p)\} \uplus \cedges(p)$ for any $p \neq r$ and $\pcedges(r) = \cedges(r)$. A neighbor of $p$ is either $p$'s parent or one of its children.
$\tree'$ being a subtree of $\tree$, denoted by $\tree' \leq \tree$, is defined as usual.

A tree $\tree = (\Proc,r,E,\edgelab)$ can be observed locally by
the processes.
Formally, the \emph{neighborhood} of a process $p \in \Proc$ is a tuple $\neigh_p = (\pedge_p, \cedges_p) \in \nset \times 2^{\nsetp}$, with the intuition that when $\tree$ is known then $\pedge_p = \edgelab(\pedge(p))$ and $\cedges_p = \edgelab(\cedges(p))$.
Given a family $\neigh = (\neigh_p)_{p \in \Proc}$ of neighborhoods, we say that the family is consistent if there exists exactly one $r \in \Proc$ such that $\pedge_r = 0$ and $\pedge_r \notin \cedges_q$ for all $q \in \Proc$, and for all $p \neq r$ there exists exactly one $q \in \Proc$ such that $\pedge_p \in \cedges_q$.
In that case, we define $\maketree{\neigh} = (\Proc,r,E,\edgelab)$ where $r$ is the unique process as defined in the previous condition, $E = \{(p,q) \mid \pedge_p \in \cedges_q\}$, and $\edgelab(p,q) = \pedge_p$.
One can easily check that if $\neigh$ is consistent, then $\maketree{\neigh}$ is a well-defined tree over $\Proc$.
Finally, we extend these notions to subsets $P \subseteq \Proc$, where
a consistent family $\neigh = (\neigh_p)_{p \in P}$ (with the $\pedge_r
= 0$ condition dropped) induces a tree $\maketree{\neigh}$ (that is a
subtree of the entire spanning tree).

\subsection{Reconfiguration operations}
We want to support reconfiguration operations to change from one TCA to another.
Fix $\Arch = (\CA,\tree)$ with the same notation as before.
There will be two kinds of operations: those affecting only the spanning tree, and those affecting only the channel memberships.
Those operations are only valid when specific conditions (that depend on which operation is involved) on the TCA are satisfied; if a condition is not satisfied then the result of that operation is undefined.

As the TCA will change due to those operations, it is more convenient to specify edges and processes in terms of integers from $\nset$, with each integer being associated to the process for which the label of the parent edge is that integer, with integer $0$ being associated to the root $r$ of the tree.
Abusing notations, we use $e$ to denote actual edges in a tree as well as integers in $\nset$, with the understanding that one goes from one to the other using the labelling bijection $\edgelab$.
With that idea in mind, we define for any tree $\tree$ the function $\procfromlab: \nset \to \Proc$ where $\procfromlab(0) = r$ and $\procfromlab(e) = p$ such that $p$ is the unique process with $\edgelab(\pedge(p)) = e$ for any $e \neq 0$.

As operations for changing the spanning tree, we define a \emph{swap} operation and a \emph{move} operation.
All operations involve a communication on some channel, say $c\in \CH$.
See Figure~\ref{example:reconfig} for an illustration.

\begin{compactitem}
\item A swap along an edge $e = (p,q)$ between a process $p$ and its parent $q$, both of which must be listening to $c$, switches their positions and makes $p$ the parent of $q$ in the new architecture.
To ensure condition~\ref{item:channelconnectivity} of Definition~\ref{def:TCA} is maintained, this operation can only be done if $p$ is listening to all channels that both $q$ and the parent of $q$ are listening to.\footnote{Alternatively, we could make $p$ connect to those channels as a result of the operation.}

Formally, we want to define the operation $\swap^c(e)(\Arch)$ where $\Arch$ is a TCA and $e \in \nsetp$.
Let $p = \procfromlab(e)$ and $q$ be the parent of $p$ in $\tree$.
For this operation to be defined, we must have $p,q \in \arch(c)$; moreover if $q$ is not the root and $q'$ is the parent of $q$ then we must have $p \in \arch(c')$ for every $c' \in \CH$ such that $q,q' \in \arch(c')$.

If those conditions are satisfied, then we define $\swap^c(e)(\Arch) = (\CA,\tree')$ with $\tree' = (\Proc,r',E',\edgelab')$ as follows.
If $q = r$ then $r' = p$ and both $E'$ and $\edgelab'$ are unchanged.
Otherwise $r' = r$, $E' = E \setminus \{(q,q')\} \uplus \{(p,q')\}$, and $\edgelab'(p,q') = \edgelab(q,q')$, $\edgelab'(e) = \edgelab(e)$ for all other edges $e$.

% --- Large step version --- %
%\item A move of $p$ to $q$ changes the position of $p$ in the tree so that its new parent after the operation is $q$.
%For this operation to be legal there are three conditions required. Let $q'$ be the parent of $p$ before the operation.
%First, $q$ can not be a descendent of $p$.
%Second, $q$ and $q'$ must be in $\Proc_c$ (but $p$ does not need to be as well).
%Finally, for all channels $c'$ such that $p$ and $q'$ are in $\Proc_{c'}$, all processes in the path from $q'$ to $q$ ($q$ included) must also belong to $\Proc_{c'}$. \mlnote{Alternatively, have all of those connect to $c'$ as a result of the operation}
%If all those conditions are satisfied, we then define $\move_c(\A,p,q) = (((\Proc_c)_{c \in \CH},T')$ where $T' = (\Proc,r,E \setminus \{(q',p)\} \uplus \{(q,p)\})$.
% --- Small step version --- %
\item A move of $p$ to $q$ changes the position of a non-root $p$ in the tree so that the new parent of $p$ after the operation is $q$.
For this operation to be legal there are three conditions required. Let $q'$ be the parent of $p$ before the operation.
First, $q$ must be a neighbor of $q'$, i.e. $p$ can only be moved to either the parent of its parent or to one of its siblings.
Second, $q$ and $q'$ must both be part of the communication on $c$ (but
$p$ need not be part of it).
Finally, moving $p$ to $q$ must not disconnect any channel that was shared by $p$ and $q'$.

Formally, the operation $\move^c(e,e')(\Arch)$ with $e \in \nsetp$ and $e' \in \nset$ is defined for $e \neq e'$ when, with $p = \procfromlab(e)$, $q = \procfromlab(e')$, and $q'$ being the parent of $p$ in $\tree$:
1) $q'$ is a neighbor of $q$,
2) $q,q' \in \arch(c)$, and
3) $q \in \arch(c')$ for all $c'$ such that $p,q' \in \arch(c')$.
Then we define $\move^c(e,e')(\Arch) = (\CA,\tree')$ where $\tree' = (\Proc,r,E \setminus \{(p,q')\} \uplus \{(p,q)\},\edgelab')$ such that $\edgelab'(p,q) = \edgelab(p,q')$ and the rest is unchanged.
\end{compactitem}

Consider the tree-like communication architecture in
Figure~\ref{fig:archi}.
Under this architecture,
$\swap^{c_1}(1)$ is possible as both $p_1$ and $p_2$ are connected to
$c_1$ and $p_1$ is the root. The only swap not possible in this
architecture is $\swap^{c_3}(4)$. Indeed, this operation would leave
$c_2$ and $c_1$ disconnected.
The operation $\move^{c_1}(3,2)$ is possible as it leaves $p_4$
connected to $c_2$.
On the other hand, operation $\move^{c_3}(3,4)$ is not possible
since it would result in $c_2$ being disconnected around $p_5$.

\begin{figure}[bt]
	\tikzset{every state/.style={minimum size=15pt}}
	\begin{center}
		\begin{tikzpicture}[
			% Default arrow tip
			%-&gt;,&gt;=stealth',shorten &gt;=1pt,
			auto,
			% Default node distance
			node distance=1.5cm,
			% Edge stroke thickness: semithick, thick, thin
			semithick,
			% Global scale
			scale=1
			]
			\begin{scope}[shift={(0,0)}]
				\node[state] (p) {$p$};
				\node[state] (q) [above right =  12mm and 4mm of p] {$q$};
				\node[] (p1) [below left =  6mm and 1mm of p] {$\square$};
				\node[] (p2) [below right = 6mm and 1mm of p] {$\triangle$};
				\node[] (q1) [below right =  6mm and 1mm of q] {$\bullet$};

				\path[-,color=red] (p) edge node [above left] {$e$} (q);
				\path[-] (p) edge (p1);
				\path[-] (p) edge (p2);
				\path[-] (q) edge (q1);
				\draw[dotted] (q) -- +(0,0.5);

				\draw[<->] (1.7,0.5) to node [above] {$\swap(e)$} +(1,0);
				\begin{scope}[shift={(3.2,0)}]
					\node[state] (q) {$q$};
					\node[state] (p) [above right =  12mm and 4mm of q] {$p$};
					\node[] (p1) [below = 5mm of p] {$\square$};
					\node[] (p2) [below right =  6mm and 2mm of p] {$\triangle$};
					\node[] (q1) [below = 5mm of q] {$\bullet$};

					\path[-,color=red] (p) edge node [above left] {$e$} (q);
					\path[-] (p) edge (p1);
					\path[-] (p) edge (p2);
					\path[-] (q) edge (q1);
					\draw[dotted] (p) -- +(0,0.5);

					\draw[-] (1.8,-1.5) -- +(0,3.5);
					\begin{scope}[shift={(2.6,0)}]
						\node[state] (p) {$p$};
						\node[state] (q') [above right =  11mm and 3mm of p] {\scalebox{0.8}{$q'$}};
						\node[state] (q) [below right =  11mm and 3mm of q'] {$q$};
						\node[] (p1) [below left =  6mm and 1mm of p] {$\square$};
						\node[] (p2) [below right =  6mm and 1mm of p]
						{$\triangle$};
						\node[] (q1) [below = 0.4cm of q] {$\bullet$};

						\path[-,color=red] (p) edge node [above left] {$e$} (q');
						\path[-] (p) edge (p1);
						\path[-] (p) edge (p2);
						\path[-] (q) edge (q1);
						\path[-] (q) edge node [above right] {$e'$} (q');
						\draw[dotted] (q') -- +(0,0.6);

						\draw[<->] (2.4,0.5) to node [above] {$\move(e,e')$} +(1,0);
						\begin{scope}[shift={(3.3,0)}]
							\node[] (i) {};
							\node[state] (q') [above right =  12mm and 3mm of i]
							{\scalebox{0.8}{$q'$}};
							\node[state] (q) [below right of=q'] {$q$};
							\node[state] (p) [below left of=q] {$p$};
							\node[] (p1) [below left =  6mm and 1mm of p] {$\square$};
							\node[] (p2) [below right =  6mm and 1mm of p]
							{$\triangle$};
							\node[] (q1) [below right =  6mm and 1mm of q]
							{$\bullet$};

							\path[-,color=red] (p) edge node [above left] {$e$} (q);
							\path[-] (p) edge (p1);
							\path[-] (p) edge (p2);
							\path[-] (q) edge (q1);
							\path[-] (q) edge node [above right] {$e'$} (q');
							\draw[dotted] (q') -- +(0,0.6);
			\end{scope}	\end{scope} \end{scope} \end{scope}
		\end{tikzpicture}
	\vspace*{-10mm}
		\caption{Illustration of $\swap$ and $\move$
		operations.}\label{example:reconfig}
	\end{center}
\vspace*{-5mm}
\end{figure}
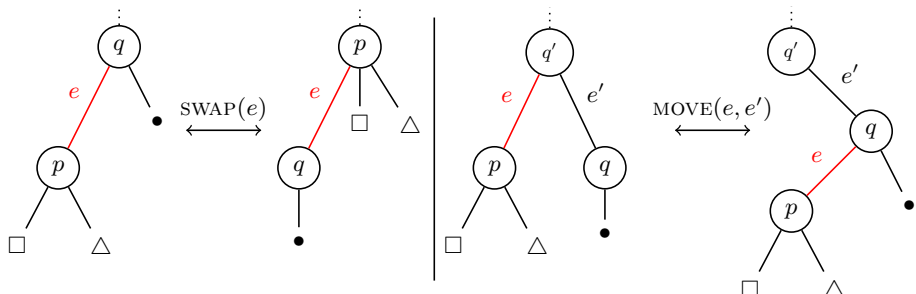

There are also two operations that change channel membership but leave the underlying spanning tree unchanged.
\begin{compactitem}
\item The \emph{connect} operation allows a process $p \in \arch(c)$ to join some specified channel $c'$.
This can happen provided that there is a process adjacent to $p$ that is contained in both $\arch(c)$ and $\arch(c')$.
Intuitively this corresponds to that process ``inviting'' $p$ to join channel $c'$ via a communication on $c$.

Operation $\connect^c(e,c')(\Arch)$ is defined for $e \in \nset$ if, with $p = \procfromlab(e)$, $p \notin \arch(c')$ and there exists $q$ a neighbor of $p$ such that $p,q \in \arch(c)$ and $q \in \arch(c')$.
Then we define $\connect^c(e,c')(\Arch) = (\arch', \tree)$ with $\arch'(c') = \arch(c) \uplus \{p\}$ and $\arch'(x) = \arch(x)$ for all $x \neq c'$.

\item Conversely, the \emph{disconnect} operation allows $p \in \Proc_c$ to stop listening to $c$, as long as it does not result in breaking the conditions of Definition~\ref{def:TCA}.

Formally, $\arch(c)$ must have size at least 3.
Furthermore, $p$ must be at the boundary of $c$, i.e. there is exactly one neighbor $q$ that is also in $\arch(c)$.
Finally, there must be a different channel $c' \neq c$ such that both $p$ and $q$ are in $\arch(c')$.
If all of these conditions are satisfied, then for $e \in \nset$ and $p = \procfromlab(e)$ we define $\disc^c(e)(\Arch) = (\arch', \tree)$ with $\arch'(c) = \arch(c) \setminus \{p\}$ and $\arch'(x) = \arch(x)$ for all $x \neq c$.
\end{compactitem}

Consider the tree-like communication architecture in
Figure~\ref{fig:archi}.
Under this architecture, it is possible to connect $p_2$ to $c_2$ by the communication $\connect^{c_1}(1,c_2)$.
It is, however, impossible to connect $p_2$ to $c_3$ as no neighbor of
$p_2$ is connected to it.
Disconnect operations are only possible for $c_1$ and $c_2$ as they
have more than $2$ processes listening to them.
However, $\disc^{c_2}(3)$ is not possible as it leaves $p_4$
disconnected.
On the other hand, $\disc^{c_2}(0)$ is possible.
Similarly, the only disconnect possible from $c_1$ is $p_3$:
$\disc^{c_1}(2)$.
See Figure~\ref{fig:connectdisconnect} for an illustration.

\begin{figure}[bt]
	\begin{center}
	\tikzset{every state/.style={minimum size=15pt}}
		\begin{tikzpicture}[
	% Default arrow tip
	%-&gt;,&gt;=stealth',shorten &gt;=1pt,
	auto,
	% Default node distance
	node distance=1.5cm,
	% Edge stroke thickness: semithick, thick, thin
	semithick
	]
	\node[state] (1) {$p_1$};
	\node[state,below left = of 1] (2) {$p_2$};
	\node[state,below right = of 1] (3) {$p_3$};
	\node[state,below left = of 3] (4) {$p_4$};
	\node[state,below right = of 3] (5) {$p_5$};

	\node[state,right = of 1,xshift=35mm] (6) {$p_1$};
\node[state,below left = of 6] (7) {$p_2$};
\node[state,below right = of 6] (8) {$p_3$};
\node[state,below left = of 8] (9) {$p_4$};
\node[state,below right = of 8] (10) {$p_5$};

	%    \tikzset{
		%      contour1/.style={red,double,double
			%distance=9mm,cap=round,},
		%      contour2/.style={blue,double,double
			%distance=10mm,cap=round,},
		%      contour3/.style={green,double,double
			%distance=11mm,cap=round,}
		%    }

	\path[->] (1) edge node[above left] {$1$} (2);
	\path[->] (1) edge node[above right] {$2$} (3);
	\path[->] (3) edge node[above left] {$3$} (4);
	\path[->] (3) edge node[above right] {$4$} (5);

	\path[->] (6) edge node[above left] {$1$} (7);
	\path[->] (6) edge node[above right] {$2$} (8);
	\path[->] (8) edge node[above left] {$3$} (9);
	\path[->] (8) edge node[above right] {$4$} (10);

	%    \begin{pgfonlayer}{bg}
		%      \draw[contour1] (2.center) --  (1.center) -- (3.center);
		%    \end{pgfonlayer}
	%    \begin{pgfonlayer}{bg}
		%      \draw[contour2] (1.center) -- (3.center) -- (4.center);
		%    \end{pgfonlayer}
	\draw[rounded corners,color=red] ($(1.north)+(0,0.4)$) --
	($(2.west)+(-0.4,0)$) -- ($(2.south)+(0,-0.4)$) --
	($(1.south)+(0,-0.4)$) -- ($(3.south)+(0,-0.4)$) --
	($(3.east)+(0.4,0)$) -- cycle node[above,color=red] {$c_1$};
	\draw[rounded corners,color=blue] ($(1.south)+(0,-0.3)$) --
	($(3.west)+(-0.3,0)$) -- ($(4.west)+(-0.3,0)$)
	node[left,color=blue] {$c_2$} -- ($(4.south)+(0,-0.3)$) --
	($(3.east)+(0.3,0)$) -- ($(1.north)+(0,0.3)$) --
	($(2.west)+(-0.3,0)$) -- ($(2.south)+(0,-0.3)$) -- cycle;
	\draw[rounded corners,color=teal] ($(3.north)+(0,0.2)$) --
	($(3.west)+(-0.2,0)$) -- ($(5.south)+(0,-0.2)$) --
	($(5.east)+(0.2,0)$) node[right,color=teal] {$c_3$} -- cycle;

	\draw[rounded corners,color=red] ($(6.north)+(0,0.4)$) --
($(7.west)+(-0.4,0)$) -- ($(7.south)+(0,-0.4)$) --
($(6.east)+(0.4,0)$) -- cycle node[above,color=red] {$c_1$};
\draw[rounded corners,color=blue] ($(6.west)+(-0.3,0)$) --
($(8.west)+(-0.3,0)$) -- ($(9.west)+(-0.3,0)$)
node[left,color=blue] {$c_2$} -- ($(9.south)+(0,-0.3)$) --
($(8.east)+(0.3,0)$) -- ($(6.north)+(0,0.3)$)  -- cycle;
\draw[rounded corners,color=teal] ($(8.north)+(0,0.2)$) --
($(8.west)+(-0.2,0)$) -- ($(10.south)+(0,-0.2)$) --
($(10.east)+(0.2,0)$) node[right,color=teal] {$c_3$} -- cycle;

\end{tikzpicture}
		\caption{The results of applying $\connect^{c_1}(1,c_2)$ and
		$\disc^{c_1}(2)$, respectively, to the
		tree-like communication architecture in Figure~\ref{fig:archi}.
	\label{fig:connectdisconnect}}
	\end{center}
\vspace*{-1cm}
\end{figure}
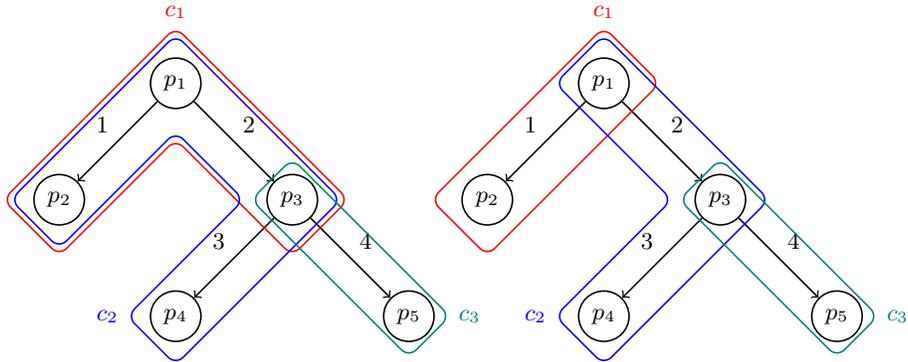

%\vspace{-0.8cm}

For any of the four introduced operations, we write $\op(\Arch)$  to
denote that $\op$ is valid from $\Arch$ and results in $\op(\Arch)$.
We show that all those operations preserve the properties of a TCA.

\begin{lemmarep}
If $\Arch$ is a TCA %and $\op$ is any of the introduced operations
then $\op(\Arch)$ is also a TCA.
\end{lemmarep}

\begin{appendixproof}
Let us check that all conditions are maintained by all possible operations.
Fix $\Arch = (\CA,\tree)$ and $c$ the channel over which the operation
is made, and let $\op(\Arch)=(\CA', \tree')$.

First, we need to check that $\tree'$ is still a valid spanning tree.
Only the first two operations change the underlying tree.
On a swap between $p$ and its parent $q$, if $q$ is the root of $\tree$ then only the root is changed and the rest of the tree does not change.
Otherwise, $q$ is not the root. Then let $q'$ be the parent of $q$.
Then the edge $(q',q)$ is removed from $\tree$ and replaced by the edge $(q',p)$.
Since by assumption there is an edge between $p$ and $q$, the resulting graph is still a spanning tree.
Similarly, on a successful move operation one edge is replaced by another such that the tree structure is preserved.

For condition~\ref{item:morethantwo} of Definition~\ref{def:TCA}, only the last two operations change channel memberships, and of them the first one can only increase the size of each set.
As for the last one, it only removes a process if there are at least two other processes in that set, so the condition is maintained.

For condition~\ref{item:channelconnectivity}, let us look at all four possible operations.
On a swap operation, the only possible way for this condition to fail is that removing the edge $(q',q)$ somehow disconnects one set $\arch(c')$.
However, this means that both $q$ and $q'$ are listening to $c'$, and thus for the operation to be valid then $p$ should also be contained in $\arch(c')$.
Thus $\arch(c')$ is still connected in the new tree, as $q$ is connected to $p$ which itself is connected to $q'$.
The argument is similar for the move operation.
For the connect operation, it can only add a new process $p$ to a given $\arch(c')$ if $p$ has a neighbor already in that same set.
Therefore the set remains connected.
Finally for the disconnect operation, the requirement is that for $p$ to disconnect from $c$ it must have exactly one neighbor in $\arch(c)$.
Thus removing $p$ does not disconnect the set $\arch(c)$.

Finally, let us show that condition~\ref{item:edgechannel} is maintained by all operations.
On a swap operation, since $\Arch$ is a TCA there is by condition~\ref{item:edgechannel} at least one channel $c'$ shared by $q$ and its parent.
If the swap is defined, then $p$ also listens to this channel.
Therefore the new edge $(q',p)$ is indeed covered by $\arch(c')$.
A similar argument works for the move operation.
The connect operation only adds a process to a channel, and so cannot falsify this condition.
Finally, on a disconnect operation, only the edge $(p,q)$ could be left without a channel covering it after $p$ disconnects from $c$.
However, by assumption both $p$ and $q$ also listen to some different channel $c'$, so the condition remains satisfied.
\qed
\end{appendixproof}

We also show that those operations are universal, meaning that one can go to any TCA using a sequence with only those four operations.

\begin{lemmarep}
If $\Arch$ and $\Arch'$ are in $\TCAset{\CH}{\Proc}$, then there is a
sequence of reconfiguration operations leading from $\Arch$ to $\Arch'$.
\end{lemmarep}

\begin{appendixproof}
	First step is to connect all processes to all channels,
	by incrementally increasing the membership of all channels by inviting neighbors to join using $\connop$ operation 
	until all processes are members.
	Then, there are no channel restrictions on moves and swaps.
	Let $r'$ be the root of $\Arch'$.
	The next step is to move $r'$ from wherever it stands initially in $\Arch$ to the root.
	This is done by a sequence of $\swapop$ operations applied to the channel
	identifying $r'$ until $r'$ becomes the root.
	Then, we move all processes to be direct children of the new root $r'$. 
	That is, for every process that is not yet a child of $r'$, we apply a $\moveop$ operation on it to push it one level above, and repeat this until all processes are children of $r'$.
	Now it is possible to move all the descendants of a process $p$ to
	be the children of $p$ by another sequence of $\moveop$ operations.
	Once the tree is arranged until level $i$ a similar procedure
	arranges level $i+1$.
	We continue all the way to the leaves so that the tree corresponds to the one in $\Arch'$.
	The final step is to disconnect processes from channels they should not listen to.
	Starting from the extremes of the tree, we disconnect one by one every process listening to a channel it should not using $\discop$ operations until we get the correct set for each channel.
\qed
\end{appendixproof}

\section{Reconfiguration languages}
\label{sec:reconfig langs}

\subsection{Automata over reconfiguration languages}
We define \emph{reconfiguration languages}, which are basically languages that also include information about reconfiguration of a TCA.
Formally, we fix $\Proc$ and $\CH$, as well as the set $\nset$ as defined before.
Then we define the set of operations $\Op = \{\nop, \swapop(e),
\moveop(e,e'), \connop(e',c), \discop(e')\}$ with $e \in \nsetp$,
$e' \in \nset$ and $c \in \CH$;
$\nop$ denotes the empty operation that does not change the TCA and is always valid (i.e. no conditions are needed for this operation).
All other operations mirror those defined previously.
Then we let $\Alp = \CH \times \Op$ be the \emph{reconfiguration
alphabet}.
A reconfiguration language is a language over this alphabet.
We note that operations as applied to TCAs include the channel name on
which they were sent in them.
We write $\op^c$ of the operation on TCA that arises from sending $\op$
on channel $c$.
For example, if $\op=\moveop(e,e')$ then
$\op^c=\move^c(e,e')$.

Now we want to define what it means for an automaton to recognize a reconfiguration language.
Intuitively, a generic DFA over this alphabet would be missing some component, namely the underlying TCA that is being affected by the reconfigurations operations.
Therefore, we introduce a new class of automata tailored to the concept of reconfiguration languages.
It can be seen as a DFA where part of the state describes the current TCA, and a transition $(c,\op)$ leads to a new state where the TCA part is the result of applying $\op$ to the previous TCA, provided such an operation is valid.

Formally a \emph{DFA for reconfiguration languages}, shortened to
RL-DFA, is a tuple $\A = (\Alp, S, \TCAs, s_0, \Arch_0, \Delta, F)$,
where $S$ is a set of states, $\TCAs$ is the set of TCA over $\CH$ and
$\Proc$ as defined earlier, $s_0 \in S$ is the initial state, $\Arch_0
\in \TCAs$ is the initial TCA, $\Delta: S \times \Alp \rightharpoonup
S$ is the partial transition function, and $F \subseteq S$ is the set
of accepting states.
As hinted, a \emph{configuration} of $\A$ is a pair $(s,\Arch) \in S
\times \TCAs$ indicating the current state and current TCA of $\A$.
In runs, we require that all transitions in $\Delta$ are valid with
respect to the current TCA, i.e. if the current TCA is $\Arch$ then
$\Delta(s,(c,\op))$ can be taken only if $\op^c(\Arch)$ is defined.
A \emph{run} of $\A$ over word
$w = (c_1,\op_1), \dots ,(c_k,\op_k)\in\Sigma^*$ is the sequence
\[
(s_0, \Arch_0) \xrightarrow{(c_1,\op_1)} (s_1, \Arch_1) \xrightarrow{(c_2,\op_2)} \dots \xrightarrow{(c_k,\op_k)} (s_k,\Arch_k)
\]
where for all $1 \leq i \leq k$, $s_i = \Delta(s_{i-1},  (c_i, \op_i))$
and $\Arch_i = \op_i^{c_i}(\Arch_{i-1})$.
Notice, that, as $\op_i(\Arch_{i-1})$ is undefined if $\op_i$ is not
possible in $\Arch_{i-1}$, this means that all operations in a run are
possible.
So, for some word the associated run does not exist if at some point
the necessary transition does not exist in the automaton or because
from the reached architecture the next operation would not be valid.
We say that configuration $(s,\Arch)$ is reachable, if there is a run
visiting it.
Word $w\in\Sigma^*$ is accepted if the run over $w$ ends in a state $s_k \in F$,
and the language accepted by $\A$, $\Lang(\A)$, is the set of words
accepted by $\A$.

\subsection{Distributivity of a language}\label{sec:distrib}
%We start from a centralized language that includes information about the reconfiguration operations, assuming that this language is recognized
%by a finite automaton.
%The goal is to build a distributed system able to recognize the same language using only the communication structure given by following the language.
We start from a centralized language given by an RL-DFA $\A$.
The goal is to build a distributed system able to recognize the same language using only the communication architecture given by following the language.
This distributed system should be in the form of an RAA $\A$, as
defined in
Section~\ref{sec:automata}.

Let us define our notion of distribution. Intuitively,
an RAA distributes an RL-DFA if both automata recognize the same language,
and if they agree, on all runs, on the sequences of communication architectures that they construct.
In runs of the centralized RL-DFA, the communication architecture occurs
explicitly in each configuration;
in runs of the distributed RAA, the communication architecture features
only implicitly, but to be able to define our notion of distribution,
we assume that the neighborhood of a process can be extracted from its
local states, using a family of functions $(N_p)_{p\in \Proc}$.
Finally, one should be able to deduce the state of the centralized automaton by looking at the product of the
local states in the RAA, using a function $D$.

Formally, let $\A = (\Alp, S, \TCAs, s_0, \Arch_0, \Delta, F)$ be a
RL-DFA and let $\A_\|=((S_p)_{p\in \Proc}, (s_p^0)_{p\in \Proc},
(\delta_p)_{p\in \Proc}, (L_p)_{p\in \Proc}, (F_p)_{p\in \Proc})$ be an
RAA over $\Alp$ (slightly adapting the definition to accept this extended alphabet).
%\dah{alphabet of RAA currently is hardcoded to be $\CH$ it seems
%(definition is just after Theorem 1); but here we intend to use
%$\CH\times\Op$ as alphabet it seems}
%\np{Yes. Good point. Might be too late to change now.}
We say that $\A_\|$ \emph{distributes} $\A$ if
there  is a family $(N_p:S_p \to \nset \times 2^{\nsetp})_{p\in \Proc}$ of functions and a function $D: \prod_{p \in \Proc} S_p \to S$ such that
for all words
$w=(c_1, \op_1),\ldots ,(c_k, \op_k)\in\Alp^*$
with
\[
\rho = (s_0, \Arch_0) \xrightarrow{(c_1, \op_1)} (s_1, \Arch_1) \xrightarrow{(c_2, \op_2)} \dots \xrightarrow{(c_k, \op_k)} (s_k, \Arch_k)
\]
being the run of $\A$ on $w$, we have that the
run of $\A_{\parallel}$ on $w$ is of shape
\[
\rho_\| = (s_p^0) \xrightarrow{(d_1,(c_1,\op_1))} (s_p^1)
\xrightarrow{(d_2,(c_2,\op_2))} \dots \xrightarrow{(d_k,(c_k,\op_k))}
(s_p^k),
\]
where for all $0\leq i\leq k$, we have
\begin{compactenum}
	\item $\Arch_i=(\left(\arch: c \mapsto \{p \in \Proc \mid c \in
	L_p(s^i_p)\}\right),\maketree{(N_p(s^i_p))_{p \in \Proc}})$,
	that is, $\Arch_i$ is the architecture arising from collecting the
	listening functions of all processes and building the tree from their
	local knowledge,
	\item $s_i=D((s^i_p)_{p\in\Proc})$, and
	\item $s_k\in F$ if and only if $\exists d \in \Data. \forall p \in
	\Proc. F_p((s^k_p),d)=\top$.
\end{compactenum}
Recall that $d_1 d_2\ldots d_k\in \Data^*$ is the data used
for coordination when communicating while reading $w$.
Finally, if the run of $\A$ on $w$ is undefined, then so is the run of
$\A_{\parallel}$.
Note that by definition $L(\A)= L(\A_{\parallel})$
whenever $\A_{\parallel}$ distributes $\A$.

\subsection{Diamond closed RL-DFA}
In general, it is not the case that any RL-DFA can be distributed.
For instance, one can create a language which only accepts the word $(c_1,\nop),
(c_2,\nop)$ for two channels $c_1$ and $c_2$ that do not share a common
process;
such a pattern can not be accepted in a distributed way as processes in
$c_2$ have no way to ``wait'' until the communication in $c_1$ happens.
A distributed language would thus have to accept also the communication $(c_2,\nop),(c_1,\nop)$.
Intuitively, for a centralized language to be distributable, the existence and result of $\Delta(s,(c,\op))$ should only depend on the history of processes participating in the communication on $c$.
In other words, if no process can distinguish between two different communications, then those two communications should end in the same configuration.
In the non-reconfiguration case, i.e. for DFA where the communication architecture is fixed, this property can be ensured by checking syntactically that the DFA has the $I$-diamond property: for any two \emph{independent} channels $a$ and $b$, meaning that there is no process listening to both $a$ and $b$, we have that $\Delta(s,ab) = \Delta(s,ba)$ for all states $s$.
Obviously this does not work in the case of RL-DFA, as two channels can be independent in one TCA but not in another; thus we need to redefine what it means to be independent.

%Given a TCA $\Arch = (\CA,\tree)$, we say that two actions $(c_1,\op_1)$ and $(c_2,\op_2)$ are $\Arch$-independent if there exists an edge $e = (p,q)$ in $\tree$ such that neither $p$ or $q$ are in $\Proc_{c_1}$ or $\Proc_{c_2}$ and such that removing $e$ from $\tree$ gives two trees $\tree_1$ and $\tree_2$ with $\op_1$ acting only on $\tree_1$ and $\op_2$ acting only on $\tree_2$.\ml{todo}
Given a TCA $\Arch = (\CA,\tree)$, we say that two actions $(c_1,\op_1)$ and $(c_2,\op_2)$ are $\Arch$-dependent if there exists $p \in \Proc$ such that $p$ listens to $c_1$ in either $\Arch$ or $\op_1^{c_1}(\Arch)$ \emph{and} process $p$ listens to $c_2$ in either $\Arch$ or $\op_2^{c_2}(\Arch)$.
Otherwise, we say that $(c_1,\op_1)$ and $(c_2,\op_2)$ are $\Arch$-independent.
In other words, $(c_1,\op_1)$ and $(c_2,\op_2)$ are $\Arch$-dependent iff there is some process $p$ that can distinguish between the two communications $(c_1,\op_1),(c_2,\op_2)$ and $(c_2,\op_2),(c_1,\op_1)$ starting from $\Arch$.
Alternatively, using the properties of TCAs, $(c_1,\op_1)$ and $(c_2,\op_2)$ are $\Arch$-independent iff there exists an edge $e$ in $\tree$ that if removed splits $\tree$ in two smaller TCA $\Arch_1 = (\CA_1, \tree_1)$ and $\Arch_2 = (\CA_2, \tree_2)$ and such that $(c_i, \op_i)$ is valid in the restricted TCA $\Arch_i$ for $i \in \{1,2\}$.
In simpler terms, those two actions are independent because they operate in disjoint parts of the tree.
It is possible to see that the two definitions are equivalent by
considering the set of processes participating in each, observing that
they are disjoint, and choosing an edge that disconnects them.

From now on, we assume that the centralized RL-DFA $\A$ is diamond
closed:
\begin{definition}[Diamond closed RL-DFA]
An RL-DFA $\A$ is said to be \emph{diamond closed} when
for all reachable configurations $(s,\Arch) \in S \times \TCAs$ and all
$\Arch$-independent actions $(c_1,\op_1)$ and $(c_2,\op_2)$,
\begin{align*}
\Delta(s, (c_1,\op_1)(c_2, \op_2)) = \Delta(s, (c_2, \op_2)
(c_1, \op_1)).
\end{align*}
\end{definition}

We extend the definition of independence to sequences, so that $w_1,w_2\in\Sigma^*$ are independent if all possible interleavings are
indistinguishable.
Formally, we define the notion recursively.
The empty sequence $\varepsilon$ is $\Arch$-independent with any
word $w$ for any TCA $\Arch$, and $w_1 = (c_1, \op_1) \cdot
w_1'$ is $\Arch$-independent with $w_2 = (c_2, \op_2) \cdot w_2'$ if
$(c_1,\op_1)$ and $(c_2,\op_2)$ are $\Arch$-independent and $w_1'$ and
$w_2$ are $\op_1^{c_1}(\Arch)$-independent and $w_1$ and $w_2'$ are
$\op_2^{c_2}(\Arch)$-independent.
If $\A$ is diamond closed with $(s,\Arch)$ being a reachable configuration, and $w_1$ and $w_2$ are two $\Arch$-independent words, then we have that $\Delta(s, w_1 w_2) = \Delta(s, w_2 w_1) = \Delta(s, w)$ for any $w$ that is an interleaving of $w_1$ and $w_2$.

We end on a property satisfied by diamond closed RL-DFA that will be used in the next section.
Fix a reachable configuration $(s,\Arch)$ and let $w_1$ and $w_2$ be two $\Arch$-independent words.
Let $s_1 = \Delta(s,w_1)$, $s_2 = \Delta(s,w_2)$, and let $C$ be the
set of channels involved in $w_2$.
We claim that one can compute the diamond closure $s' =
\Delta(s,w_1w_2)$ knowing only $s$, $s_1$, $s_2$, and $C$.

\begin{lemmarep}[\cite{KrishnaM13,cori1993asynchronous}]
Let $\A$ be a diamond closed RL-DFA. There is a function $\diam: S^3 \times 2^{\CH} \to S$ such that if $(s,\Arch), s, s_1, w_1, s_2, w_2, C, s'$ are all defined as described above, then $\diam(s,s_1,s_2,C) = s'$.
\end{lemmarep}

\begin{appendixproof}
One can compute $\diam(s,s_1,s_2,C)$ as follows: first, guess a TCA
$\Arch'$ such that $(s,\Arch')$ is reachable in $\A$.
The channels in $C$ lie in one part of $\Arch'$; let $C'$ be a set of channels located in a disjoint part of $\Arch'$.
Then guess two words $w_1'$ and $w_2'$ restricted to channels in $C'$ and $C$ respectively and such that $s_i = \Delta(s, w_i')$ for $i \in \{1,2\}$.
We know this is possible because $\Arch$, $w_1$ and $w_2$ satisfy those conditions.
Let $s'' = \Delta(s_1,w_2')$, and let us show that $s''$ is the desired output i.e. $s'' = s'$.
First, by hypothesis we have that $s' = \Delta(s,w_1w_2) = \Delta(s_2,w_1)$.
We also know that $\Delta(s,w_2') = s_2$, so $s' = \Delta(s,w_2' w_1)$.
Then, note that $w_2'$ and $w_1$ are $\Arch'$ independent due to their set of channels $C$ and $C'$ being separated in $\Arch'$.
Thus, we can use the diamond property: $s' = \Delta(s,w_2' w_1) = \Delta(s,w_1 w_2')$.
We conclude by seeing that $\Delta(s,w_1 w_2') = \Delta(s_2,w_2') = s''$.
\qed
\end{appendixproof}

%In \cite{KrishnaM13}, this function has been used in the following way.
%On a communication between a parent and one of its children, let $s$ be the last state in common in their views, $s_1$ and $s_2$ be the most up-to-date states known by the parent and its child respectively, and $C$ the set of channels belonging in the subtree rooted in the child.
%Then one can combine the information known by the parent and the child by computing $\diam(s,s_1,s_2,C)$, without needing to know the full sequences $w_1,w_2$ that were done separately.

To understand the need for this function, let us give a quick intuition for the construction in the case of a fixed tree with binary channels given in \cite{KrishnaM13}.
There, the goal of the construction is to distribute a diamond closed DFA $\A$ into an asynchronous automaton.
Each process will maintain a state of $\A$ that is the most recent state visited by $\A$ according to the knowledge of this process.
This does not necessarily coincide with the actual state of $\A$, as some communication happening without this process involved can change the state of $\A$.
Then when two processes communicate they will share their information in order to update their knowledge about the state of $\A$.
To that end, both processes need to agree on what was the state the last time they communicated with each other and what each side has been doing independently since that time.
With a tree structure, this is actually easy: each child also maintain the state of $\A$ during the last communication with their parent, and by the structure of the tree everything that happened with the child since then is independent from what happened with the parent.
So on a communication between a parent and one of its children, let $s$ be the state after their last communication, $s_1$ and $s_2$ be the most up-to-date states known by the parent and its child respectively, and $C$ the set of channels belonging in the subtree rooted in the child.
Then one can combine the information known by the parent and the child by computing $\diam(s,s_1,s_2,C)$, without needing to know the full sequences $w_1,w_2$ that were done separately on each side since $s$.

In our setting, communications can include more than two participants.
To that end, we naturally extend the $\diam$ function to work on trees.
Let $\tree$ be a tree where each node is labeled by a pair $(s^1,s^2)$ of states and a set $C$ of channels, with the intuition that $s^1$ is the last state in common that the process at this node has with its parent, $s^2$ is the most recent state this process knows, and $C$ is the set of channels that can be found in the subtree rooted in that node.
Let $p$ be the root of $\tree$, labeled by $(s^1_r,s^2_r,C_r)$, and assume that $p$ has $k$ children labeled $(s^1_i,s^2_i,C_i)$ and leading to tree $\tree_i$ for $i \in \{1,\dots,k\}$.
Then we define the function $\diamtree$ recursively in the following way:
\begin{align*}
\diamtree(\tree) =
\begin{cases}
s^2_r &\text{ if $k = 0$,}\\
\diam(s^1_1, \diamtree(\tree \setminus \tree_1), \diamtree(\tree_1), C_1) &\text{ otherwise.}
\end{cases}
\end{align*}
Basically we just iterate the $\diam$ function on all children of a parent, then repeat in all levels of the tree until we get to the root.
Note that the $\mathrm{state}(p)$ function from \cite{KrishnaM13} is equivalent to $\diamtree(\tree_p)$ where $\tree_p$ is the tree rooted in $p$.

We now have all the tools needed to state our construction.

\section{Distribution of Diamond Closed RL-DFA}
\label{sec:diamond closed}

We show that the language of a diamond closed RL-DFA can
be distributed to collaborating RAAs.
The different processes of the RAA maintain their knowledge of the
architecture as well as their knowledge regarding the states of the
RL-DFA that the computation has reached.
We then show that based on the communication that each process
participates in, they can keep this information updated so that it is
maintained correctly and corresponds to the respective configurations of
the RL-DFA.

More formally we have the following.
We fix the finite sets $\Proc$ and $\CH$ of processes
and channels, respectively, and let $n=|\Proc|$.
Given an RL-DFA $\A = (\Alp, S, \TCAs, s_0, \Arch_0, F, \Delta)$, the goal of this
section is to show how to build an RAA $\A_\|=((S_p)_{p\in \Proc},
(s_p^0)_{p\in \Proc},
(\delta_p)_{p\in \Proc}, (L_p)_{p\in \Proc}, (F_p)_{p\in \Proc})$ that
distributes $\A$.
Toward this goal, we describe the maintained information in the
states of each process.

\subsection{Local information to be maintained}
Each process $p$ stores in its states $S_p$ some local information about the communication architecture as well as some information about the state of the centralized automaton.
Recall that we associate each edge in the tree with an integer in $\nsetp$.
A state $s_p\in S_p$ then is a tuple of the form $\langle (s^1, s^2), L, (\pedge, \cedges), \cc, \dc \rangle$
where:
\begin{compactitem}
\item $(s^1, s^2) \in S\times S$ is a pair of states of the centralized
automaton, with $s^1$ being the state during the last communication
involving both $p$ and its parent, and $s^2$ the most up-to-date state
known by $p$.
\item $L \subseteq \CH$ is the set of channels that $p$ listens to, i.e. $L_p(s_p) = L$.
In the corresponding TCA, it means that $c \in L$ iff $p \in \arch(c)$.
\item $(\pedge, \cedges) \in \nset \times 2^{\nsetp}$ are the edge leading to the parent of $p$ (or $0$ for the root) and the set of edges leading to the children of $p$ respectively.
The edge $\pedge$ serves as a proxy for the identity of $p$ and its position in the tree.
Let $\pcedges$ denote the set $\{\pedge\} \uplus \cedges$.
\item $\cc: \pcedges \to 2^L$ is the set of \emph{connected channels} shared by $p$ and each of its neighbors.
For instance, $\cc(\pedge)$ is exactly the set of channels that both $p$ and its parent listen to.
As we want to maintain a TCA, these sets should never be empty by condition~\ref{item:edgechannel} of Definition~\ref{def:TCA}.
\item $\dc: \pcedges \to 2^{\CH \setminus L}$ indicates for each direction leading away from $p$ which set of channels can be found in that direction that $p$ does not listen to ($\dc$ stands for \emph{disconnected channels}).
Intuitively, we have $c \in \dc(e)$ for some $e \in \pcedges$
if all processes that listen to channel $c$ can be found in the subtree obtained by following the edge whose label is $e$.
By condition~\ref{item:channelconnectivity} of Definition~\ref{def:TCA}, any channel $c\in \CH\setminus L$ can only be found in at most one direction, otherwise $\arch(c)$ would be disconnected due to $p$ not being contained in this set; and by condition~\ref{item:morethantwo}, $c$ can be found in at least one direction.
Thus $\dc$ can also be seen as a partition of $\CH \setminus L$ over $\pcedges$.

%\dhnote{can we always be sure that every channel from $\CH\setminus L$ occurs in $\dc$? Perhaps make $\dc$ have type $\CH\setminus L\to \pcedges$?}
\end{compactitem}

The first component, that is the pair of states of the centralized automaton, is identical to the construction in \cite{KrishnaM13}.
%That is, the first state of this pair is the state reached during the last communication with the parent of $p$, while the second state is the most up-to-date state $p$ knows of.
The rest of the components are used to represent a TCA in a distributed way and help maintain it under reconfiguration operations.

Recall that we fixed $\A = (\Alp, S, s_0, F, \Delta)$ and an initial TCA $\Arch_0 = (\CA_0,\tree_0)$.
Let $\tree_0 = (\Proc, r_0, E_0, \edgelab_0)$.
By abuse of notation, we extend $\edgelab_0$ to set of edges.
The initial state of process $p \neq r_0$ is
\[s_p^0 = \langle (s_0, s_0), \arch_0^{-1}(p), (\edgelab_0(\pedge(p)), \edgelab_0(\cedges(p))),
\cc_p^0, \dc_p^0 \rangle\]
where $\cc_p^0$ and $\dc_p^0$ are defined as follows:
\[
\begin{array}{l c l}
\cc_p^0(e) & = & \{c \in \CH \mid p,q \in \arch_0(c) \text{ where $q$
is the neighbor of $p$ along edge $\edgelab^{-1}(e)$}\} \\
\dc_p^0(e) & = &  \left \{
c \in \CH \left |
\begin{array}{l} p \notin \arch_0(c) \land \text{
there is a path in $\tree_0$ from $p$ starting}\\
\text{with edge $\edgelab^{-1}(e)$ reaching some $q$ with $q \in
\arch_0(c)$}
\end{array} \right . \right \}
\end{array}
\]
For $p = r_0$, component $\edgelab_0(\pedge(p))$ (which is undefined)
is set to $0$ instead; all other components are defined same as above
with $\cc_p(0) = \dc_p(0) = \emptyset$.

%The next section introduces the invariants that will be satisfied throughout a run of $\A_\|$.
We are now ready to start the description of the transition function $\Delta_p$ of each process $p$.
Transitions will vary based on which reconfiguration operation occurs, and the data sent along the communication will reflect this.
However, there is some basic common data that is needed for synchronization, even when no reconfiguration happens.
We start by describing this data in the next section.

\subsection{Common synchronization data}
Under every communication, processes participating in the communication
share as part of the data the information they have regarding the
subtree of processes participating in the communication.
Namely, on a communication in channel $c$ only the processes in
$\arch(c)$ participate.
The shared data will allow every process to independently build the
subtree of $c$'s participants,
check that there is no inconsistency in the data sent, and then compute
the newest state by taking the diamond closure over the states sent by
all other processes in the correct order.
Based on this computation, and potentially requiring additional data,
each process updates their state.
%We note that if a process cannot validate the consistency of the data
%they block the communication. It follows that the only communications
%occurring are those with consistent data.

Formally, the data sent is of the following form:
\begin{align*}
	&\SyncData = \{((s_1^1,s_1^2),\pedge_1,\cedges_1,C_1), \dots,
	((s_k^1,s_k^2),\pedge_k,\cedges_k,C_k) \mid \\
	&\quad 2 \leq k \leq n\ \land\ \forall 1 \leq i \leq k. (s_i^1,s_i^2)
	\in S^2, \pedge_i \in \nset, \cedges_i \subseteq \nsetp, C_i \subseteq \CH\}
\end{align*}
A process $p$ agrees only to data of this form that is consistent
both with $p$'s local view and with itself; otherwise $p$ will simply
block this communication.
Consider a data $\sync =
((s_1^1,s_1^2),\pedge_1,\cedges_1,C_1), \dots,
((s_k^1,s_k^2),\pedge_k,\cedges_k,C_k)  \in \SyncData$, we define
$\treesync = \maketree{(\pedge_i, \cedges_i)_{1 \leq i \leq k}}$.
Each node of $\treesync$ is associated to one index $i$ between $1$ and $k$, and is labeled by $(s_i^1,s_i^2,C_i)$.
For all processes $p \in \Proc$, given $p$'s a local state $s_p =
\langle (s^1, s^2),
L, (\pedge, \cedges), \cc, \dc \rangle$, we say that data
\begin{equation*}
\sync =
((s_1^1,s_1^2),\pedge_1,\cedges_1,C_1), \dots,
((s_k^1,s_k^2),\pedge_k,\cedges_k,C_k)  \in \SyncData
\end{equation*}
is \emph{consistent} with $s_p$ for channel $c$, if $c
\in L$ and the following conditions hold:
\begin{compactitem}
	\item The subtree $\treesync$ is defined.
	\item There are no two $i \neq j$ such that $\pedge_i = \pedge_j$.
	\item There is $i \leq k$ such that $(s_i^1,s_i^2) =
	(s^1,s^2)$, $\pedge_i = \pedge$, $\cedges_i = \{e \in \cedges \mid c
	\in \cc(e)\}$, and for all $e \in \cedges_i$, there exists $j \leq k$ such that
	$\pedge_j = e$ and $C_j = \dc(e)$.
%	\item If $r \leq k$ is such that $\pedge_r$ is the root of $T_c$, then $C_r = \CH$.
\end{compactitem}
We denote by $\statefromsync(\sync) = \diamtree(\treesync)$ the newest
up-to-date state of the centralized automaton $\A$ according to
processes that listen to $c$ agreeing on data $\sync$.
It follows that the information in the synchronization data is
sufficient for computing the diamond closure for all the
processes participating in the communication.

We describe the transition function for all possible reconfiguration
operations.

\subsection{Transition on reconfiguration operation}
\paragraph*{With no reconfiguration operation ($\nop$).}
Under a $(c, \nop)$ communication, the TCA remains unchanged.
In the distributed automaton, this is reflected in that only the first
component of each local state needs to be updated.
As the TCA is unchanged, no additional information is required in
addition to the synchronization data.

Consider $p$'s local state $s_p =
\langle (s^1, s^2),
L, (\pedge, \cedges), \cc, \dc \rangle$.
Let $\statefromsync(\sync)=\diamtree(\treesync)$.
Then, $\delta_p(s_p, (\sync, (c, \nop)))$
is defined if and only if $\sync$ is consistent with $s_p$ for $c$ and $s' =
\Delta(\statefromsync(\sync), (c,\nop))$ is
defined.
We put
\[
\begin{array}{c}
\delta_p(s_p, (\sync, (c, \nop)))=
\langle
(\idifroot{c}{s_p}{s'},
s'), L, (\pedge, \cedges), \cc, \dc \rangle, \mbox{ where}\\[8pt]

	\idifroot{c}{s_p}{s'} =
\left \{
\begin{array}{l r}
		s^1 &\text{ if $p$ is the root of the subtree of $c$ (i.e. $c
		\notin \cc(\pedge)$),}\\
		s' &\text{ otherwise.}
		\end{array}\right .
	\end{array}
\]

\paragraph*{On $\swapop(e)$.}
A communication $(c, \swapop(e))$ along an edge $e = (p,q)$ swaps the position of a non-root process $p$ and its parent $q$ in the TCA.
Note that $p$ is identified by being the unique process whose parent edge is $e$, and similarly $q$ is the unique process for which $e$ is a child edge.
Recall that this operation is only possible if $p$ listens to every
channel that is shared by $q$ and the parent of $q$, so as to not break
connectivity after swapping.
Thus the distributed automaton should also ensure that this condition is satisfied; this requires additional information sent in the data during communication.
If this condition is satisfied, then $p$ and $q$ swap their upward edges and update their $\cc$ and $\dc$ components accordingly.

Formally, for any process $\pi$ with local state $s_\pi = \langle (s^1, s^2), L, (\pedge, \cedges), \cc, \dc \rangle$
such that $c \in L$,
$\delta_\pi$ is defined on communication $(c, \swapop(e))$ with $e \neq 0$ and with data of the form $d = (\sync, C, D) \in \SyncData \times (2^{\CH})^2$ if:
\begin{compactitem}
\item $\sync$ is consistent with $s_\pi$ for $c$,
\item the edge $e$ must be part of the subtree $\treesync$,
\item if $e \in \cedges$, $C$ is the set of channels such that $C =
\cc(\pedge)$, that is $C$ is the set of channels shared by $q$ and its
parent, and $D = \dc(\pedge)$,
\item if $e = \pedge$, $C \subseteq L$, that is $p$ listens to every
channel in $C$,
\item $s' = \Delta(\statefromsync(\sync), (c, \swapop(e)))$ must be defined in $\A$.
\end{compactitem}
If those conditions are fulfilled, let $q$ be the unique index in $\sync$ such that $e \in \cedges_q$ (which every process can compute given the data in $\sync$).
Then we define $\delta_\pi(s_\pi, (d, (c, \swapop(e))))$ as follows:

\begin{tabular}{| l | c | c | c |}
\hline
new state ... & if $e = \pedge$ (process $p$) & if $e \in \cedges$ (process $q$) & otherwise \tabularnewline
\hline
$(s'^1,s'^2)$ & \multicolumn{3}{ c |}{$(\idifroot{c}{s_\pi}{s'}, s')$}\tabularnewline
\hline
$L'$ & \multicolumn{3}{ c |}{$L$}\tabularnewline
\hline
$(\pedgenew, \cedgesnew)$ & $(\pedge_q, \cedges \uplus \{e\})$ & $(e, \cedges \setminus \{e\})$ & $(\pedge, \cedges)$ \tabularnewline
\hline
$\cc'$ & $\cc'(\pedge_q) = C$ & $\cc'(e) = \cc(e)$ & $\cc$ \tabularnewline
 &~ $\forall e \in \cedgesnew. \cc'(e) = \cc(e) $ ~&~ $\forall e \in \cedgesnew. \cc'(e) = \cc(e) $ ~& \tabularnewline
\hline
$\dc'$ & $\dc'(\pedge_q) = D$ & $\dc'(e) = \dc(e) \cup D$ & $\dc$ \tabularnewline
 & $\dc'(e) = \dc(e) \setminus D$ &~ $\forall e \in \cedgesnew. \dc'(e) = \dc(e) $ ~& \tabularnewline
 &~ $\forall e \in \cedges. \dc'(e) = \dc(e)$ ~&  & \tabularnewline
 \hline
\end{tabular}

\paragraph*{On $\moveop(e,e')$.}
A $\moveop(e,e')$ changes the position of the process $p$ whose parent edge is $e$ to become one of $q$'s children (where $q$ is the process whose parent edge is $e'$), assuming that $q$ was a neighbor of $p$'s original parent, which we call $q'$.
Another condition is that any channel shared by $p$ and its original parent $q'$ must also be shared by its new parent $q$; this set of channels needs to be sent as part of the data as in the previous operation.
Notably, $p$ itself does not need to be part of the communication, as nothing changes in the TCA from its point of view.
As for process $q$, it needs to add $e$ to the set of its children edges, and know what set of channels are found in $p$'s subtree to update its state correctly; those will be provided in the communication by $q'$.
Process $q'$ itself removes $e$ from its children edges, and
changes its $\dc$ to reflect $p$'s new position.

Let $\pi$ be any process listening to $c$ and let its local state be $s_\pi = \langle (s^1, s^2), L,\\
(\pedge,\cedges), \cc, \dc \rangle$.
Then the conditions for $\delta_\pi(s_\pi, (c, \moveop(e,e')), d)$ to be defined with data $d = (\sync, C) \in \SyncData \times 2^{\CH}$ are:
\begin{compactitem}
\item $\sync$ is consistent with $s_\pi$ for $c$,
\item edge $e'$ must be part of the subtree $\treesync$ (but not
necessarily $e$),
\item if $e \in \cedges$, $C = \cc(e)$ is the set of channels shared by
$p$ and its parent,
\item if $e' = \pedge$, $C \subseteq L$, i.e. $q$ listens to every
channel in $C$,
\item $s' = \Delta(\statefromsync(\sync), (c, \moveop(e,e')))$ must be defined in $\A$.
\end{compactitem}
Let $D = C_{q'}$ from $\sync$ where $q'$ is the unique index such that $e \in \cedges_{q'}$, and let $e_{qq'}$ be the edge joining $q$ and $q'$.
Then we have $\delta_\pi(s_\pi, (d, (c, \moveop(e,e'))))$:

\begin{tabular}{| l | c | c | c |}
\hline
new state & if $e' = \pedge$ (process $q$) & if $e \in \cedges$ (process $q'$) & otherwise \tabularnewline
\hline
$(s'^1,s'^2)$ & \multicolumn{3}{ c |}{$(\idifroot{c}{s_\pi}{s'}, s')$}\tabularnewline
\hline
$L'$ & \multicolumn{3}{ c |}{$L$}\tabularnewline
\hline
$(\pedgenew, \cedgesnew)$ & $(\pedge, \cedges \uplus \{e\})$ & $(\pedge, \cedges \setminus \{e\})$ & $(\pedge, \cedges)$ \tabularnewline
\hline
$\cc'$ & $\cc'(e) = C$ & $\forall e \in \pcedgesnew. \cc'(e) = \cc(e) $ & $\cc$ \tabularnewline
 & $\forall e \in \pcedges. \cc'(e) = \cc(e)$ &  & \tabularnewline
\hline
$\dc'$ & $\dc'(e) = D$ & $\dc'(e_{qq'}) = \dc(e_{qq'}) \cup D$, & $\dc$ \tabularnewline
 & $\forall e \in \pcedges. \dc'(e) = \dc(e)$ & $\forall e \in \pcedgesnew \setminus \{e_{qq'}\}. \dc'(e) = \dc(e) $ & \tabularnewline
 \hline
\end{tabular}

\paragraph*{On $\connop(e,c')$.}
On a $\connop(e,c')$ communication sent on channel $c$, and with $p$ being the process whose parent edge is $e$, $p$ starts listening to $c'$.
This requires the existence of a process, say $q$, that is a neighbor of $p$ and that listens to both $c$ and $c'$.
To check this, an additional edge will be sent as part of the data, with that edge representing process $q$ (which will block the communication if it does not fit the required conditions).
As a result of this, $p$ and $q$ will have to update their $\cc$ and $\dc$ components.

Again, let $\pi$ be a process listening to $c$ and with local state $s_\pi = \langle (s^1, s^2), L,\\ 
(\pedge,\cedges), \cc, \dc \rangle$.
Let $d = (\sync, e') \in \SyncData \times \nset$, we have that $\delta_\pi(s_\pi, (c, \\
\moveop(e,c')), d)$ is defined if:
\begin{compactitem}
\item $\sync$ is consistent with $s_\pi$ for $c$,
\item edges $e$ and $e'$ are neighbors in the subtree $\treesync$,
\item if $e = \pedge$ then $c' \notin L$,
\item if $e' = \pedge$ then $c' \in L$,
\item $s' = \Delta(\statefromsync(\sync), (c, \connop(e,c')))$ must be defined in $\A$.
\end{compactitem}
With $p$ being the process represented by $e$ and $q$ by $e'$, let $e_{pq}$ be the edge joining $p$ and $q$, which is either $e$ if $q$ is the parent of $p$ and $e'$ otherwise.
Then we have $\delta_\pi(s_\pi, (d, (c, \moveop(e,c'))))$ as follows:

\begin{tabular}{| l | c | c | c |}
\hline
new state & if $e = \pedge$ (process $p$) & if $e' = \pedge$ (process $q$) & otherwise \tabularnewline
\hline
$(s'^1,s'^2)$ & \multicolumn{3}{ c |}{$(\idifroot{c}{s_\pi}{s'}, s')$}\tabularnewline
\hline
$L'$ & $L \uplus \{c'\}$ & \multicolumn{2}{ c |}{$L$}\tabularnewline
\hline
$(\pedgenew, \cedgesnew)$ & \multicolumn{3}{ c |}{$(\pedge, \cedges)$} \tabularnewline
\hline
$\cc'$ & \multicolumn{2}{ c |}{$\cc'(e_{pq}) = \cc(e_{pq}) \uplus \{c'\}$} & $\cc$ \tabularnewline
 & \multicolumn{2}{ c |}{$\forall e \in \pcedges \setminus \{e_{pq}\}. \cc'(e) = \cc(e) $}  & \tabularnewline
\hline
$\dc'$ & $\dc'(e_{pq}) = \dc(e_{pq}) \setminus \{c'\}$ & \multicolumn{2}{ c |}{$\dc$} \tabularnewline
 &~ $\forall e \in \pcedges \setminus \{e_{pq}\}. \dc'(e) = \dc(e)$ ~& \multicolumn{2}{ c |}{} \tabularnewline
 \hline
\end{tabular}

\paragraph*{On $\discop(e)$.}
Finally, a $\discop(e)$ operation on channel $c$ makes the process $p$ whose parent edge is $e$ disconnect from $c$.
Several conditions need to be satisfied for this to be possible.
First, exactly one neighbor $q$ of $p$ must be in $c$, which can be checked from the tree $T_c$.
There must be at least two processes other than $p$ listening to $c$, which again can be seen in the size of $T_c$.
Lastly, there must be at least one other channel $c'$ shared by $p$ and $q$; both of those can check this by looking at their $\cc$ component.
If those are satisfied, then $p$ disconnects fron $c$ and $p$ and $q$ update their $\cc$ and $\dc$.

Let $\pi$ be a process listening to $c$ and let its local state be $s_\pi = \langle (s^1, s^2), L, (\pedge,\\
\cedges), \cc, \dc \rangle$.
$\delta_\pi(s_\pi, (c, \moveop(e,e')), d)$ is defined with data $d = \sync \in \SyncData$ if:
\begin{compactitem}
\item $\sync$ is consistent with $s_\pi$ for $c$,
\item $e$ is in $\treesync$ and:
\begin{compactenum}
\item Either $e$ is the parent edge of a leaf, then let $e_{pq} = e$ and we must have that there exists some $c' \neq c$ such that $c' \in \cc(e_{pq})$ if $e_{pq} \in \pcedges$.
\item Or $e$ is the parent edge of the root and the root has only one children edge $e'$, then let $e_{pq} = e'$ and we require the same condition as above.
\end{compactenum}
% is either the parent edge of a
%leaf or the parent edge of the root and the root has only one children
%edge,
\item $|\treesync| \geq 3$,
%\item Let $e_{pq}$ be either $e$ if $e$ is the parent edge of a leaf in $T_c$, or
%$e'$ if $e$ is the parent edge of the root and $e'$ is the unique children edge of the root (i.e. $e_{pq}$ is the edge connecting $p$ and $q$).
%Then there exists some $c' \neq c$ such that $c' \in \cc(e_{pq})$ if $e = \pedge$,\np{Please clarify this sentence}\ml{Better now?}
\item $s' = \Delta(\statefromsync(\sync), (c, \discop(e)))$ must be defined in $\A$.
\end{compactitem}
With $e_q$ being the parent edge of the only neighbor of $p$ in $\treesync$, we define $\delta_\pi(s_\pi, (d, (c, \moveop(e,c'))))$ as:

\begin{tabular}{| l | c | c | c |}
\hline
new state & if $e = \pedge$ (process $p$) & if $e_q = \pedge$ (process $q$) & otherwise \tabularnewline
\hline
$(s'^1,s'^2)$ & \multicolumn{3}{ c |}{$(\idifroot{c}{s_\pi}{s'}, s')$}\tabularnewline
\hline
$L'$ & $L \setminus \{c\}$ & \multicolumn{2}{ c |}{$L$}\tabularnewline
\hline
$(\pedge, \cedgesnew)$ & \multicolumn{3}{ c |}{$(\pedge, \cedges)$} \tabularnewline
\hline
$\cc'$ & \multicolumn{2}{ c |}{$\cc'(e_{pq}) = \cc(e_{pq}) \setminus \{c\}$} & $\cc$ \tabularnewline
 & \multicolumn{2}{ c |}{$\forall e \in \pcedges \setminus \{e_{pq}\}. \cc'(e) = \cc(e) $}  & \tabularnewline
\hline
$\dc'$ & $\dc'(e_{pq}) = \dc(e_{pq}) \uplus \{c\}$ & \multicolumn{2}{ c |}{$\dc$} \tabularnewline
 &~ $\forall e \in \pcedges \setminus \{e_{pq}\}. \dc'(e) = \dc(e)$ ~& \multicolumn{2}{ c |}{} \tabularnewline
 \hline
\end{tabular}

\paragraph{Data contents.}
With $\delta_p$ now defined on all operations, we can formally state what is the set of data contents:
\[\Data = \SyncData \cup (\SyncData \times (2^\CH)^2) \cup (\SyncData \times 2^\CH) \cup (\SyncData \times [n])\]

\subsection{Acceptance and correctness}
The last component that remains to be defined is the acceptance
condition $F_p$.
To that end, we define \emph{global} synchronization data, which is basically some synchronization data $\sync$ that includes all processes instead of being restricted to only those participating in a communication.
Formally, we say that
\[
\sync =
((s_1^1,s_1^2),\pedge_1,\cedges_1,C_1), \dots,
((s_n^1,s_n^2),\pedge_n,\cedges_n,C_n)  \in \SyncData
\]
is globally consistent with $s_p$ if:
\begin{compactitem}
\item $\treesync$ is defined,
\item $\forall i. \pedge_i = i-1$, i.e. each process is seen once in order of their parent edge,
\item with $i = \pedge_p + 1$, we have $s_i^1 = s_p^1, s_i^2 = s_p^2, \pedge_i = \pedge_p, \cedges_i = \cedges_p$,
\item $\forall e \in \cedges_p. C_{e+1} = \dc_p(e)$.
%\item with $s = \diamtree(\treesync)$, we have that $s \in F$.
\end{compactitem}
We then put $F_p(s_p,\sync) = \top$ if and only if $\sync$ is globally consistent with $s_p$ and $\statefromsync(\sync) \in F$.

Let $\A_p = (S_p, s_p^0, \delta_p, L_p, F_p)$ as defined throughout
this section and let
$\A_\|=((S_p)_{p\in \Proc},
(s_p^0)_{p\in \Proc},
(\delta_p)_{p\in \Proc}, (L_p)_{p\in \Proc}, (F_p)_{p\in \Proc})$.
We state our main result as follows.
\begin{theorem}\label{thm:correctness}
If $\A$ is diamond closed, then $\A_\|$ distributes $\A$.
\end{theorem}

The rest of this section is dedicated to the proof of Theorem~\ref{thm:correctness}.

First, and completely independently from the construction, we need to define what is the view of a process and then show a useful property relating $\diamtree$ and those views.
Let $\A$ be a diamond closed RL-DFA, $w \in \Alp^\ast$ be a word, $\rho = (s_0, \Arch_0) \xrightarrow{(c_1,\op_1)} \dots \xrightarrow{(c_k,\op_k)} (s_k,\Arch_k)$ the run of $\A$ on $w$, and $X \subseteq \Proc$ a set of process.
As processes in $X$ are not necessarily part of every communication in $w$, there may be some parts of $\rho$ that are happening outside of those in $X$'s knowledge and they may not know the correct final state reached at the end of $\rho$ just by themselves.
We formally define the \emph{view} of $X$ on $w$, which intuitively correspond to the most up-to-date that processes in $X$ can deduce at the end of $\rho$, and denoted by $\traceview{X}{w}$, recursively on $\rho$.
\begin{align*}
&\traceview{X}{\varepsilon} = s_0\\
&\traceview{X}{w \cdot (c_k,\op_k)} =
\begin{cases}
\Delta(\traceview{X \cup \arch_{k-1}(c_k)}{w}, (c_k,\op_k)) &\text{if $X \cap \arch_{k-1}(c_k) \neq \emptyset$,}\\
\traceview{X}{w} &\text{otherwise.}
%\traceview{X}{w} &\text{if $X \cap \arch_{k-1}(c_k) = \emptyset$,}\\
%\Delta(\traceview{X \cup \arch_{k-1}(c_k)}{w}, (c_k,\op_k)) &\text{otherwise.}
\end{cases}
\end{align*}
That is, starting from the last communication in which at least one in $X$ was part of, $\traceview{X}{w}$ computes backward the state that is reached through actions in which someone in $X$ has participated, then $X$ and those that shared a communication with someone in $X$, and then those that shared a communication with those, and so on.
Naturally, we have that $\traceview{\Proc}{w} = \Delta(s_0,w)$ for any word $w$.
Furthermore, with $\A$ being diamond closed, we can deduce that for any action $(c,\op)$, $\Delta(s_k,(c,\op))$ is defined if and only if $\Delta(\traceview{\arch(c)}{w},(c,\op))$ is defined.
That is, whether a communication in channel $c$ can happen cannot depend on actions made outside of the view of those listening to $c$.
We also define, for some process $p \in \Proc$, its \emph{shared parent view}, denoted by $\traceparentview{p}{w}$ as the view of $p$ and its parent in $\Arch_{i-1}$ if $1 \leq i \leq k$ is the maximal index for which both are listening to $c_i$, or $s_0$ is no such index exists.

The $\diam$ function, and by extension $\diamtree$, interplay nicely with the views of independent parts of the system.
More specifically, we show that if each node in $\tree$ is correctly labeled, then $\diamtree(\tree)$ correctly computes the view according to every process in $\tree$.

\begin{lemma}\label{lemma:diamview}
For all words $w$ with the run on $w$ ending in configuration $(s,\Arch)$, if $\tree$ is a subtree of $\Arch$ where each node $p$ is labeled by $(\traceparentview{p}{w}, \traceview{p}{w}, C_p)$, then $\diamtree(\tree) = \traceview{\{p \mid p \in \tree\}}{w}$.
\end{lemma}

\begin{proof}
By induction on the structure of $\tree$, same as in \cite{KrishnaM13}.
If $\tree$ is empty then $\diamtree(\tree) = s_0 = \traceview{\emptyset}{w}$.
Otherwise, assume $\tree$ is a root $p$ labeled by $(s^1_p, s^2_p,\CH)$ with $k$ children labeled by $(s^1_1,s^2_1,C_1), \dots, (s^1_k,s^2_k,C_k)$ and leading to subtrees $\tree_1, \dots, \tree_k$ respectively.
Assume that the property holds on those subtrees, i.e. that for all $i \leq k$ $\diamtree(\tree_i) = \traceview{\{p \mid p \in \tree_i\}}{w}$.
Then with $T_i'$ denoting the subtree of $\tree$ composed of $p$ and its $i$ first children, we recursively show that
\[
P_i: \diamtree(\tree_i') = \traceview{\{p \mid p \in \tree_i'\}}{w}
\]
For the initial step $i = 0$, $\diamtree(\tree) = s^2_p = \traceview{\{p\}}{w}$ by assumption.
Now suppose $P_{i-1}$ holds for some $1 \leq i \leq k$.
Then we have $\diamtree(\tree_i') = \diam(s^1_i, \diamtree(\tree_{i-1}'), \diamtree(\tree_i), C_i)$, and we know that $s^1_i = \traceparentview{i}{w}$, that $\diamtree(\tree_{i-1}') = \traceview{\{p \mid p \in \tree_i'\}}{w'}$, and that $\diamtree(\tree_i) = \traceview{\{p \mid p \in \tree_i\}}{w}$.

Let $w'$ be the maximal prefix of $w$ such that $\traceparentview{i}{w} = \traceparentview{i}{w'}$, i.e. the last action in $w'$ is the last common action between the root of $\tree_i$ and its parent at that time (which might or might not be $p$ due to later reconfiguration operations, but that is irrelevant).
Then let $w''$ be such that $w = w' \cdot w''$, and let $w_i$ be the minimal subword of $w''$ containing all actions involving channels in $C_i$.
Since after $w'$ there are no communications between $\tree_i$ and the rest of the tree, necessarily $\traceview{\{p \mid p \in \tree_i\}}{w} = \traceview{\{p \mid p \in \tree_i\}}{w' \cdot w_i}$.
As $w_i$ is chosen to be minimal, we can deduce that $\traceview{\{p \mid p \in \tree_i\}}{w' \cdot w_i} = \Delta(\traceview{\{p \mid p \in \tree_i\}}{w'}, w_i)$,
and thus that $\diamtree(\tree_i) = \Delta(s^1_i, w_i)$.
Then let $w_{i-1}$ be the rest of $w''$, i.e. the word such that $w''$ is an interleaving of $w_i$ and $w_{i-1}$.
Similarly, we obtain that $\diamtree(\tree_{i-1}') = \Delta(s^1_i, w_{i-1})$.
Finally, by definition of the $\diam$ function, we get that $\diamtree(\tree_i') = \diam(s^1_i, \diamtree(\tree_{i-1}'), \diamtree(\tree_i), C_i) = \Delta(s^1_i, w_i \cdot w_{i-1}) = \Delta(s^1_i, w'') = \traceview{\{p \mid p \in \tree_i'\}}{w}$.
We have shown that $P_i$ holds for all $i$, and therefore for $i = k$, and as $\tree= \tree_k'$ we have $\diamtree(\tree) = \traceview{\{p \mid p \in \tree\}}{w}$.
\qed
\end{proof}

Note that this lemma implies that if $\tree$ is the full tree and is correctly labeled, then $\diamtree(\tree) = \traceview{\Proc}{w} = \Delta(s_0,w)$.
We will use this property to prove the correctness of our construction.

Let us define invariants that should be satisfied throughout a run.
Let $\A$ be a diamond closed RL-DFA and $\A_\|$ be as defined in Section~\ref{sec:diamond closed}.

% ---------- Statement of invariants ---------- %
Let $w=(c_1, \op_1)(c_2, \op_2)\ldots (c_n, \op_n)\in\Alp^*$.
Our first invariant states that a run of $\A_\|$ on $w$ exists if and only if a run of $\A$ exists.
\[
\invdef{w}: \text{The run $\rho$ of $\A$ on $w$ is defined $\Leftrightarrow$ The run $\rho_\|$ of $\A_\|$ on $w$ is defined.}
\]
Now if both runs are undefined, then $w$ is trivially rejected by both $\A$ and $\A_\|$ and there is nothing left to show.
So for the remaining invariants, assume that both $\rho$ and $\rho_\|$ are defined:
\begin{align*}
\rho &= (s_0, \Arch_0) \xrightarrow{(c_1, \op_1)} (s_1, \Arch_1) \xrightarrow{(c_2, \op_2)} \dots \xrightarrow{(c_n, \op_n)} (s_n, \Arch_n)\\
\rho_\| &= (s_p^0) \xrightarrow{d_1,(c_1,\op_1)} (s_p^1) \xrightarrow{d_2,(c_2,\op_2)} \dots \xrightarrow{d_n,(c_n,\op_n)} (s_p^n)
\end{align*}
with $\Arch_n = (\CA, \tree)$ and $s_p^n = \langle (s^1_p, s^2_p), L_p, (\pedge_p,
\cedges_p), \cc_p, \dc_p \rangle$.
We define the function $N$ that is simply the projections on the third component of a state, i.e. $N(\langle
(s^1, s^2), L, (\pedge, \cedges), \cc, \dc \rangle) = (\pedge, \cedges)$.
Then we establish the invariants relating those two runs.
The first pair of invariants relate the actual state of $\rho$ with the pair of states that each process keeps in their local state, which corresponds to the intuition given in the previous section.
\begin{align*}
&\invstateparent{w}: \forall p \in \Proc. s^1_p = \traceparentview{p}{w}
&\invstate{w}: \forall p \in \Proc. s^2_p = \traceview{\{p\}}{w}
\end{align*}
Those two invariants, used with Lemma~\ref{lemma:diamview}, are key to prove that the languages are equivalent.
Then we have invariants that relate to the TCA itself.
\begin{align*}
&\invCA{w}: \forall c \in \CH. \CA(c) = \{p \in \Proc \mid c \in L_p\}
&\invtree{w}: \tree = \maketree{(N(s_p))_{p \in \Proc}}
\end{align*}
Namely, $\invCA{w}$ states that at after the run the set of
channels process $p$ is connected to is the same according to the
distributed and centralized view.
Invariant $\invtree{w}$ states that the centralized view of the
spanning tree in the architecture is the composition of the local
environments maintained by individual processes.
\begin{align*}
&\invcc{w}: \forall p \in \Proc, e \in \nset. \cc_p(e) = \\
&\quad \{c \in \CH \mid p,q \in \arch(c) \text{ where $q$ is the neighbor of $p$ along edge $\edgelab^{-1}(e)$}\}\\
&\invdc{w}: \forall p \in \Proc, e \in \nset. \dc_p(e) = \{c \in \CH \mid\ p \notin \arch(c) \land \text{ there is a path in $\tree$}\\
&\quad \text{from $p$ starting with edge $\edgelab^{-1}(e)$ reaching some $q$ with $q \in \arch(c)$}\}
\end{align*}
Those two invariants basically state that the $\cc$ and $\dc$ component are exactly as described.
Let us now prove that all those invariants hold inductively.

% ---------- Proof of invariants ---------- %
% ----- Init ----- %
\paragraph*{Initialization of the invariants.}
All invariants hold on the empty word $w = \varepsilon$.
$\invdef{\varepsilon}$ is trivial.
$\invstateparent{\varepsilon}, \invstate{\varepsilon}, \invCA{\varepsilon}, \invtree{\varepsilon}, \invcc{\varepsilon},$ and $\invdc{\varepsilon}$ are all easily derived from the definition of the initial state $s_p^0$.

% ----- nop ----- %
\paragraph*{Invariants after a $\nop$ operation.}
Now consider a word of the form $w' = w \cdot (c,\nop)$, with runs
\begin{align*}
\rho &= (s_0, \Arch_0) \xrightarrow{w} (s, \Arch) \xrightarrow{(c,\nop)} (s',\Arch')\\
\rho_\| &= (s_p^0) \xrightarrow{w} (s_p) \xrightarrow{\sync, (c,\nop)} (s_p')
\end{align*}
Assume that all invariants hold in $w$.
We show that they still hold in $w'$.

For $\invdef{w'}$, first note that there is exactly one possible $\sync$ (up to reordering) that is consistent with $s_p$ for all processes $p$ listening to $c$, which is the tuple where each of these processes share their information correctly.
The order of the tuple does not matter, as the tree $\treesync$ and the diamond closure $\statefromsync(\sync)$ are computed in an order-agnostic way.
Then we have that $s_p' = \delta_p(s_p, (\sync, (c, \nop)))$ is defined for all $p \in \arch(c)$ if and only if state $s' = \Delta(\statefromsync(\sync), (c,\nop))$ is defined, and due to $\invdef{w}$ we know that the runs on $w$ were both defined.
Thus $\rho$ is defined if and only if $\rho_\|$ is defined.
Let us assume both are.

For $\invstateparent{w'}$ and $\invstate{w'}$, let $p \in \Proc$.
Let $s_p = \langle (s^1_p, s^2_p), L_p, (\pedge_p,\cedges_p), \cc_p, \dc_p \rangle$
and $s_p' = \langle ((s^1)_p', (s^2)_p'), L_p', ((\pedge)_p',(\cedges)_p'), \cc_p', \dc_p' \rangle$.
If $p$ was not part of the communication on $c$, then $s_p = s_p'$ and by definition $\traceparentview{p}{w'} = \traceparentview{p}{w}$ and $\traceview{\{p\}}{w'} = \traceview{\{p\}}{w}$, so $\invstateparent{w'}$ and $\invstate{w'}$ are trivially obtained from $\invstateparent{w}$ and $\invstate{w}$ respectively.
Assume now that $p$ was part of this communication.
Note that by $\invCA{w}$ and $\invtree{w}$, we know that $\treesync$ is exactly the subtree containing the processes in $\arch(c)$.
Let us first show $\invstate{w'}$.
By definition, $\traceview{\{p\}}{w'} =\Delta(\traceview{\arch(c)}{w},(c,\nop))$.
By combining $\invstate{w}$, $\invstateparent{w}$, $\invdc{w}$, and Lemma~\ref{lemma:diamview}, we get that $\traceview{\arch(c)}{w} = \diamtree(\treesync) = \statefromsync(\sync)$.
Thus $\traceview{\{p\}}{w'} =\Delta(\statefromsync(\sync),(c,\nop)) = (s^2)_p'$.
For $\invstateparent{w'}$, there are two cases.
First, assume that the parent of $p$ in $\tree$ is not part of the communication on $c$ (either because $p$ is the root, or because its parent does not listen to $c$).
In that case, $\traceparentview{p}{w'} = \traceparentview{p}{w} = s^1_p$ by definition of $\traceparentview{p}{w'}$ and by $\invstateparent{w}$.
Moreover, $p$ must be the root of $\treesync$.
Therefore $(s^1)_p' = s^1_p$ and $\invstateparent{w'}$ is satisfied.
Second, assume now that $p$'s parent, say $q$, is part of the communication.
Then by definition $\traceparentview{p}{w'} = \traceview{\{p,q\}}{w'} = \Delta(\traceview{\arch(c)}{w},(c,\nop))$, which is $(s^2)_p'$ by the same proof as for $\invstate{w'}$.

Invariants $\invCA{w'}, \invtree{w'}, \invcc{w'}, \invdc{w'}$ all trivially hold because nothing related to the TCA has been modified by the last transition.

% ----- swap ----- %
\paragraph*{Invariants after a $\swapop(e)$ operation.}
Let $w' = w \cdot (c,\swapop(e))$, with associated runs
\begin{align*}
\rho &= (s_0, \Arch_0) \xrightarrow{w} (s, \Arch) \xrightarrow{(c,\swapop(e))} (s',\Arch')\\
\rho_\| &= (s_p^0) \xrightarrow{w} (s_p) \xrightarrow{(\sync,C,D),(c,\swapop(e))} (s_p')
\end{align*}

For $\invdef{w'}$, let us first assume that $\rho$ is defined, and so that $(c,\swapop(e))$ is valid from $\Arch$.
By $\invdef{w}$, $\rho_\|$ is defined up to $w$.
Just as with the $\nop$, there exists a (unique up to reordering) $\sync$ data that is consistent with the state of every process for $c$.
Since the $\swapop$ is valid in $\Arch$, then by $\invCA{w}$ and $\invtree{w}$ exists exactly one $p$ and $q$ such that $\pedge_p = e$ and $e \in \cedges_q$, both of which are listening to $c$.
Therefore $e$ is in $\treesync$.
The set $C = \cc(\pedge_q)$ is exactly the set of channels shared by $q$ and its parent (if it exists) by $\invcc{w}$.
Since $\swapop$ is valid, then $p$ must listen to all of $C$, so the condition $C \subseteq L_p$ is satisfied.
Finally, as $\rho$ is defined, then necessarily $s' = \Delta(s, (c,\swapop(e)))$ is defined.
From this we have that $\Delta(\traceview{\arch(c)}{w}, (c,\swapop(e)))$ is defined, and as $\traceview{\arch(c)}{w} = \diamtree(\treesync) = \statefromsync(\sync)$, the last condition is satisfied and so $\rho_\|$ is defined.
The reverse direction can be shown similarly.

The proof of $\invstateparent{w'}$ and $\invstate{w'}$ is exactly the same as with the $\nop$ (and that holds true with all remaining operations, so we shall skip those two invariants from now on).

$\invCA{w'}$ is derived from $\invCA{w}$ as the swap does not change either $\CA$ or the listening functions: $L' = L$ for all processes.

For $\invtree{w'}$, after the transition we have $((\pedge_p)',(\cedges_p)') = (\pedge_q,\cedges_p \uplus \{e\})$ and $((\pedge_q)',(\cedges_q)') = (e,\cedges_q \setminus \{e\})$, and for the other processes nothing changes.
So the new tree $\tree' = \maketree{(N(s_p'))_{p \in \Proc}}$ is exactly $\tree$ but with the positions of $p$ and $q$ swapped, which is the result of applying $\swapop(e)$ to $\Arch$.

For $\invcc{w'}$, let us first look at process $p$.
We set $\cc'_p(\pedge_q) = C$: we know that $p$ is listening to all of $C$ otherwise the transition would not be defined and that the process which has $\pedge_q$ as one of its children (again, if it exists) is also listening to all of $C$ because $C = \cc_q(\pedge_q)$ and by $\invcc{w}$.
Moreover, there cannot be a channel $c' \notin C$ that both $p$ and $q$'s parent both listen to but not $q$, as it would break the connectivity condition of the TCA.
Thus, $C$ is exactly the set of channels shared by $p$ and its new parent.
Since the other neighbors of $p$ are the same, the rest of $\cc_p$ is unchanged.
Similarly process $q$ only lost one neighbor, with $p$ becoming its new parent, so nothing needs to change in $\cc_q$.
Thus $\invcc{w'}$ holds.

Lastly, for $\invdc{w'}$ let us first look again at $p$ first.
By assumption $D = \dc_q(\pedge_q)$, so since $p$ is swapping with $q$ we set $\dc_p'((\pedge_p)') = \dc_p'(\pedge_q) = D$.
Since by $\invdc{w}$ those channels in $D$ were previously found in $\dc_p(e)$, we need to remove them in $\dc_p'(e)$.
Otherwise the rest is unchanged.
The same reasoning for $q$ states that the channels in $D$, that were previously found in $\dc_q(\pedge_q)$ can now be found in $\dc_q'(e)$, as $e$ leads to $p$ which leads to $q$'s previous parent.
No other process has anything to change, and so $\invdc{w'}$ holds.

% ----- move ----- %
\paragraph*{Invariants after a $\moveop(e,e')$ operation.}
Let $w' = w \cdot (c,\moveop(e,e'))$, with associated runs
\begin{align*}
\rho &= (s_0, \Arch_0) \xrightarrow{w} (s, \Arch) \xrightarrow{(c,\moveop(e,e'))} (s',\Arch')\\
\rho_\| &= (s_p^0) \xrightarrow{w} (s_p) \xrightarrow{(\sync,C),(c,\moveop(e,e'))} (s_p')
\end{align*}

For $\invdef{w'}$, let us first assume that $\rho$ is defined, and so that $(c,\moveop(e,e'))$ is valid from $\Arch$.
By $\invdef{w}$, $\rho_\|$ is defined up to $w$.
Just as with the $\nop$, there exists a (unique up to reordering) $\sync$ data that is consistent with the state of every process for $c$.
Since the $\moveop$ is valid in $\Arch$, then by $\invCA{w}$ and $\invtree{w}$ exists exactly one $p$, $q$, and $q'$ such that $\pedge_p = e$, $\pedge_q = e'$, and $e \in \cedges_{q'}$, with both $q$ and $q'$ listening to $c$.
Since $q$ listens to $c$, then $e'$ is in $\tree_c$.
The set $C = \cc(e)$ is exactly the set of channels shared by $p$ and $q'$, by $\invcc{w}$.
Since the move is valid, then $q$ must also listen to all of $C$, so the condition $C \subseteq L$ if $\pedge = e'$ is satisfied.
Finally, as $\rho$ is defined, then necessarily $s' = \Delta(s, (c,\move(e,e')))$ is defined and with a similar reasoning as with the previous operation we get that $\Delta(\statefromsync(\sync),  (c,\move(e,e')))$ is defined.
%Again, we leave the reverse direction to the diligent reader.

$\invCA{w'}$ holds because $\invCA{w}$ does and because the move operation changes neither $\CA$ nor the listening functions of each process.

As for $\invtree{w'}$, only the neighborhoods of $q$ and $q'$ are modified by the move operation.
More specifically, we have that $(\cedges_q)' = \cedges_q \uplus \{e\}$ and $(\cedges_{q'})' = \cedges_{q'} \setminus \{e\}$, corresponding exactly to the transfer of $p$ from $q'$ to $q$.

For $\invcc{w'}$ and $\invdc{w'}$, $q$ has to incorporate the information regarding $e$ by copying the old information of $q'$ (given by $C$ and $D$ for $\cc_q$ and $\dc_q$ respectively).
For $q'$, it simply removes $e$ from its $\cc$ component because $e$ is not a neighbor anymore, and everything that was in $\dc(e)$ is now transferred to the edge linking $q$ and $q'$, since $e$ can now be found in that direction.
Everything else is unchanged, and one can verify that $\invcc{w'}$ and $\invdc{w'}$ both hold.

% ----- conn ----- %
\paragraph*{Invariants after a $\connop(e,c')$ operation.}
Let $w' = w \cdot (c,\connop(e,c'))$, with associated runs
\begin{align*}
\rho &= (s_0, \Arch_0) \xrightarrow{w} (s, \Arch) \xrightarrow{(c,\connop(e,c'))} (s',\Arch')\\
\rho_\| &= (s_p^0) \xrightarrow{w} (s_p) \xrightarrow{(\sync,e'),(c,\connop(e,c'))} (s_p')
\end{align*}

For $\invdef{w'}$, let us first assume that $\rho$ is defined, and so that $(c,\connop(e,c'))$ is valid from $\Arch$.
By $\invdef{w}$, $\rho_\|$ is defined up to $w$.
Just as with the $\nop$, there exists a (unique up to reordering) $\sync$ data that is consistent with the state of every process for $c$.
Since the $\connop$ is valid in $\Arch$, $p$ does not listen to $c'$ and there is a neighbor $q$ of $p$ that does, and both $p$ and $q$ listen to $c$.
Take $e'$ to be the parent edge of any such $q$ (actually there can only be at most one due to connectivity, but that is not relevant yet).
Then by $\invCA{w}$ and $\invtree{w}$ we have that $\pedge_p = e$, $\pedge_q = e'$, and with $p$ and $q$ both being in $\treesync$.
Also $c' \notin L_p$ and $c' \in L_q$.
Finally, as $\rho$ is defined, then necessarily $s' = \Delta(s, (c,\connop(e,c')))$ is defined and with again a similar reasoning we get that $\Delta(\statefromsync(\sync),  (c,\connop(e,c')))$ is defined.

For $\invCA{w'}$, we note that $\CA'$ is $\CA$ but with $\arch'(c') = \arch(c') \uplus \{p\}$.
By $\invCA{w}$, we have that $\arch(c') = \{\pi \in \Proc \mid c' \in L_\pi\}$.
The only change in listening functions after the transition is for $p$: $L_p' = L_p \uplus \{c'\}$.
Therefore, $\arch'(c') = \arch(c') \uplus \{p\} = \{\pi \in \Proc \mid c' \in L_\pi\} \uplus \{p\} = \{\pi \in \Proc \mid c' \in L_\pi'\}$, and so we get $\invCA{w'}$.

$\invtree{w'}$ is trivially obtained from $\invtree{w}$ as no neighborhood is changed by the transition.

For $\invcc{w'}$, as only $p$ connects to $c'$ and the only neighbor of $p$ that is also connected to $c'$ is $q$, only those two need to update their $\cc$ component to add $c'$ in each other's direction.
This is exactly what is done in $\cc'(e_{pq}) = \cc(e_{pq}) \uplus \{c'\}$, and so $\invcc{w'}$ holds.

Finally for $\invdc{w'}$ again there is only one change compared to the previous situation: now that $p$ also listens to $c'$ then it must be removed from its $\dc$ component.
Since we know that $q$ listened to $c'$ already before, then necessarily we had that $c' \in \dc_p(e_{pq})$.
Thus we remove it in $\dc_p'(e_{pq})$ and obtain the correct result, so $\invdc{w'}$ is satisfied.

% ----- disc ----- %
\paragraph*{Invariants after a $\discop(e)$ operation.}
Let $w' = w \cdot (c,\discop(e))$, with associated runs
\begin{align*}
\rho &= (s_0, \Arch_0) \xrightarrow{w} (s, \Arch) \xrightarrow{(c,\discop(e))} (s',\Arch')\\
\rho_\| &= (s_p^0) \xrightarrow{w} (s_p) \xrightarrow{\sync,(c,\discop(e)))} (s_p')
\end{align*}

For $\invdef{w'}$, let us first assume that $\rho$ is defined, and so that $(c,\discop(e))$ is valid from $\Arch$.
By $\invdef{w}$, $\rho_\|$ is defined up to $w$.
Just as with the $\nop$, there exists a (unique up to reordering) $\sync$ data that is consistent with the state of every process for $c$.
Since the $\discop$ is valid in $\Arch$, we have that $\arch(c)$ is of size at least 3, that $p$ listens to $c$ and has exactly one neighbor $q$ that also does, and that both $p$ and $q$ listen to another channel $c' \neq c$.
Then by $\invCA{w}$ and $\invtree{w}$ we have that $\pedge_p = e$ which is in $\treesync$, and as $p$ has only one neighbor in $\treesync$ then it must either be a leaf or the root of $\treesync$ with only one child.
In both cases, $e_{pq}$ is the edge joining $p$ and $q$.
Then since both $p$ and $q$ listen to $c'$, then by $\invcc{w}$ $c' \in \cc_p(e_{pq})$ and $c' \in \cc_q(e_{pq})$.
The condition $|\treesync| \geq 3$ is obtained from the condition on the size of $\arch(c)$.
Finally, as $\rho$ is defined, then necessarily $s' = \Delta(s, (c,\discop(e)))$ is defined and with again a similar reasoning we get that $\Delta(\statefromsync(\sync),  (c,\discop(e)))$ is defined.

For $\invCA{w'}$, similarly with the previous operation, here the only difference after the transition is the removal of $p$ from $\arch(c)$, which is mimicked by removing $c$ from $L_p$.
$\invtree{w'}$ is again trivial since the tree is not changed.

Again, mirroring the previous operation, for $\invcc{w'}$ and $\invdc{w'}$ the only changes this time is to remove $c$ from the channels shared between $p$ and $q$ and adding it to the $\dc_p$ in the direction of $q$, which is exactly what is done in the construction.

\paragraph*{Conclusion of the proof.}
We have shown that the invariants hold.
Let us end by showing that $\A_\|$ distributes $\A$ by showing the three conditions of Section~\ref{sec:distrib} are satisfied.

Item~1 is covered by invariants $\invCA{w}$ and $\invtree{w}$.
Combining them with Lemma~\ref{lemma:diamview}, we get that $\diamtree(\tree) = \traceview{\Proc}{w} = \Delta(s_0,w)$ when $\tree$ is correctly labelled with the information stored in the states of each process.
This tree is exactly the tree $\treesync$ where $\sync$ is the unique data globally consistent with all processes $p \in \Proc$.
Thus Item~2 is satisfied: $D$ is the function that first computes the unique data $\sync$ globally consistent with all $s_p$, then computes the corresponding tree $\treesync$, and finally applies the $\diamtree$ function to it.
Item~3 also follows as for all processes to agree on acceptance the only possible data choice is this unique $\sync$, and then all processes accept if and only if $\diamtree(\treesync)$ is in $F$.
This concludes the proof of correctness of this construction. 

\subsection{Complexity and non-reconfigurable case}
We analyze the complexity of the proposed construction.
A state of a process $p$ is of the form $s_p = \langle (s^1, s^2), L, (\pedge, \cedges), \cc, \dc \rangle$, and a global state $\textbf{s} = (s_p)_{p \in \Proc}$ is a product of the states for each process.
The first component is of the shape $(s^1, s^2) \in S^2$, so a state has quadratic dependency on the size of $S$.
Next, $L$ is a subset of $\CH$, $\pedge$ and $\cedges$ combined are a subset of $[n]$ (recall that $n = |\Proc|$), and $\cc$ and $\dc$ are partitions of $\CH$ over $[n]$.
Combining all those together, we get that $S_p$ has size quadratic in $S$ and exponential in both $\CH$ and $\Proc$.
Since in our setting $n$ is fixed, a global state $\textbf{s}$ takes
space logarithmic in $S$ and polynomial in $\CH$ and $\Proc$.
Notice that the set $\TCAset{\CH}{\Proc}$ of all possible TCAs is also
exponential in $\CH$ and $\Proc$, and as $\A$ must be diamond closed,
its state set must relate to this set.
Therefore, overall, we get the following result.
\begin{theorem}
The size of $\A_\|$ is polynomial in the size of $\A$.
\end{theorem}

Our construction is also a generalization of the construction of Krishna and Muscholl \cite{KrishnaM13}.
That is, assume no reconfiguration operation happens during a run over $w$ (i.e. we only have $\nop$ operations in $w$).
In the run of the centralized automaton $\A$, this means that the TCA never changes, and remains $\Arch_0$ throughout the run.
Similarly in the run of $\A_\|$, only the first component $(s^1,s^2)$ in each process state gets updated; the remaining components stay as they were initially.
We can then build a simpler $\A_\|$ by removing those fixed parts and
only keeping the pair of states, dropping the complexity down to simply
polynomial in $S$.
This new construction corresponds to the one from~\cite{KrishnaM13}, but generalized to tree-like architectures instead of trees in which all channels are binary.
If we further restrict channels to be binary in the now fixed architecture, we get back exactly the same construction.
Indeed if $\CA(c) = \{p,q\}$ with $p$ being the parent of $q$ then
$\diamtree(\tree_c)$ reduces to $\diam(s^1_q, s^2_p, s^2_q, C_q)$ which
is the way updates are computed in Krishna and Muscholl's construction.

\section{Conclusion}
\label{sec:conc}
We define automata over reconfiguration languages, along with a notion
of distribution by reconfigurable asynchronous automata, as well as a semantic (confluence) property that ensures distributability.
Our construction imposes a tree-like communication architecture, but supports both communication channels with an arbitrary number of members, as well as reconfiguration of the communication architecture, and therefore significantly extends the construction from \cite{KrishnaM13} that is restricted to fixed architectures with binary channels that are arranged in a tree. Our distribution method retains the quadratic dependency on
the number of states of the central automaton from~\cite{KrishnaM13},
but uses additional data for the distribution of the (changing)
communication architecture.

Future works include generalizing this result to general communication
topologies.
There are several difficulties involved.
First, devising the set of atomic reconfiguration operations that will
allow to gradually transform an architecture into any possible other architecture.
These operations need to be ``locally safe'', that is allow processes
to update their view of the architecture so that the global
architecture can still be captured and that enough information about
independence is maintained locally.
Second, the diamond property of the centralized automaton becomes more
involved and requires careful analysis of how independence changes over
time and operations.
Finally, follow up on Zielonka's construction in its full complexity
and extend it with reconfiguration.
We need to both add and update the information that processes maintain about
the architecture, as well as use this information in order to understand
(and update others about) the progress of the global computation.

\clearpage
\bibliographystyle{splncs04}
\bibliography{bib}

\end{document}